\documentclass[10pt,conference]{IEEEtran}
\usepackage{graphicx}
\usepackage{epsfig}
\usepackage{amsmath,amsfonts,dsfont,amssymb,stmaryrd}
\newtheorem{definition}{Definition}
\newtheorem{theorem}{Theorem}
\newtheorem{proposition}{Proposition}
\newtheorem{lemma}{Lemma}

\newcommand{\mR}{\mathbb{R}}

\newcommand{\Sc}{\mathcal{S}}
\newcommand{\eps}{\varepsilon}
\newcommand{\wtil}{\widetilde}
\newcommand{\lno}{\left\lVert}
\newcommand{\rno}{\right\rVert}

\newcommand{\vnl}{\stackrel{\text{VN}}{\longrightarrow}}

\newcommand{\X}{\mathcal{X}}
\newcommand{\Y}{\mathcal{Y}}
\newcommand{\Z}{\mathcal{Z}}

\newcommand{\A}{\mathcal{A}}
\newcommand{\Pn}{P_N}

\begin{document}

\title{Linear Universal Decoding for Compound Channels: \\ a Local to Global Geometric Approach}

\author{
\begin{tabular}{c c}
    \begin{minipage}{0.5\linewidth}
        \begin{center}
        Emmanuel Abbe \\
                 Massachusetts Institute of Technology \\
               Laboratory for Information and Decision Systems\\
                Cambridge, MA 02139 \\
                eabbe@mit.edu \\
        \end{center}
    \end{minipage}
    &
    \begin{minipage}{0.5\linewidth}
        \vspace{-2pt}
        \begin{center}
        Lizhong Zheng \\
                 Massachusetts Institute of Technology\\
               Laboratory for Information and Decision Systems\\
              Cambridge, MA 02139 \\
                 lizhong@mit.edu \\
        \end{center}
    \end{minipage}
    \\
\end{tabular}
}

\maketitle



\begin{abstract}
Over discrete memoryless channels (DMC), linear decoders (maximizing additive metrics) afford several nice properties. 
In particular, if suitable encoders are employed, the use of decoding algorithm with manageable complexities is permitted. Maximum likelihood is an example of linear decoder. For a compound DMC, decoders that perform well without the channel's knowledge are required in order to achieve capacity. Several such decoders have been studied in the literature. 
However, there is no such known decoder which is linear. Hence, the problem of finding linear decoders achieving capacity for compound DMC is addressed, and it is shown that under minor concessions, such decoders exist and can be constructed. \\
This paper also develops a {\it local geometric analysis}, which allows in particular, to solve the above problem. By considering very noisy channels, the original problem is reduced, in the limit, to an inner product space problem, for which insightful solutions can be found. The local setting can then provide counterexamples to disproof claims, but also, it is shown how in this problem, results proven locally can be ``lifted'' to results proven globally.
\end{abstract}

\section{Introduction}

We consider a discrete memoryless channel with input alphabet $\X$
and output alphabet $\Y$. The channel is described by the
probability transition matrix $W$, each row of which is the 
conditional distribution of the output symbol $Y$ conditioned on a
particular input $X=x\in \X$. We are interested in the compound
channel, where the exact value of $W$ is not known, either at the
transmitter or the receiver. Such problems can often be motivated by
the wireless applications with unknown fading realizations. Here,
instead of assuming the channel $W$ to be known at the receiver and transmitter, we 
assume that a set $S$ of possible channels is known at the receiver and transmitter; and
our goal is to design encoders and decoders that support reliable
communication, no matter which channel in $S$ actually takes place.

Compound channels have been extensively studied in the literature.
In particular, Blackwell et.al. \cite{BBT} shown that the highest achievable rate is given by the following expression:
\begin{eqnarray}
C(S) \stackrel{\Delta}{=} \max_{P} \inf_{W \in S} I(P,W), \label{compound}
\end{eqnarray}
where the maximization is over all probability distributions $P$ on $\X$.
Thus, $C(S)$ is referred to as the {\it compound channel
capacity}. To achieve the capacity, i.i.d. (or fixed composition) random codes from the
optimal input distribution, i.e. the distribution maximizing \eqref{compound}, are used. The random coding argument is
commonly employed to prove achievability for a single given channel,
such as in Shannon's original paper. By showing that the error
probability averaged over the random ensemble can be made
arbitrarily small, one can conclude that there exists ``good" codes
with low enough error probability. This argument is strengthened in
\cite{BBT} to show that with the random coding argument,  we can indeed prove the existence of codes that are good for all
possible channels. Adopting this view, in this paper, we will not be
concerned about constructing the code, or even finding the optimal
input distribution, but rather simply assume that one of the above 
mentioned universally good code is used, and focus on the designs of
efficient decoding algorithms.

In \cite{BBT}, a decoder that maximizes a uniform mixture of
likelihoods over most possible channels is used, and shown to achieve
capacity. 
The most general universal decoder is the
maximum mutual information (MMI) decoder \cite{ciskon}, which
computes the empirical mutual information between each codeword and
the received word and picks the highest one. The practical
difficulty of implementing MMI decoders is obvious. As empirical
distributions are used in computing the ``score" of each codeword, it
becomes challenging to efficiently store the exponentially many
scores, and update the scores as symbols being received
sequentially. Conceptually, when the empirical distribution of the
received signals is computed, one can in principle estimate the
channel $W$, making the assumption of lack in channel knowledge less
meaningful. There has been a number of different universal decoders
proposed in the literature, including the LZ based algorithm
\cite{lapziv}, or merged likelihood decoder \cite{lapfed}. In this paper, we try to find universal decoders in a 
class of particularly simple decoders: {\it linear decoders}.

Here, linear (or additive) decoders are defined to have the
following structure. Upon receiving the $n$-symbol word
$y$, the decoder compute a score/decoding metric $d^n(x_m, y)$ (note
that the score of a codeword does not depend on other codewords) for
each codeword $x_m, m=1, 2,\ldots, 2^{nR}$, and decodes to the one
codeword with the highest score (ties can be resolved arbitrarily). Moreover, the $n$-symbol decoding
metric has the following additive structure
\begin{eqnarray*}
d^{n}(x_m, y) = \sum_{i=1}^n d(x_m(i), y(i))
\end{eqnarray*}
where $d:\X \times \Y \to \mathbb{R}$ is a (single-letter) decoding
metric. Such decoders are called linear since the decoding metric
it computes is indeed linear in the joint empirical distribution
between the codeword and the received word, since
\begin{eqnarray*}
d^{n}(x_m, y) = n \cdot \sum_{a\in \X,b\in \Y}
\hat{P}_{(x_m,y)}(a,b) \cdot d(a,b)
\end{eqnarray*}
where $\hat{P}_{(x_m,y)}$ denotes the joint empirical distribution
of $(x_m, y)$. We call such a decoder a linear decoder induced by-
the single-letter metric $d$.

Linear decoders have been widely studied in \cite{cisnar,lapmis}.
An additive decoding metric has some obvious advantages. First, when
used with appropriate codes, it allows the decoding complexity to be
reduced. Note that maximum likelihood (ML) decoder is by definition
a linear decoder, with single-letter metric $d=\log W$, the log likelihood of the
channel, thus linear decoders can potentially use the existing
decoder structures to simplify designs. For example, when
convolutional codes are used, Viterbi algorithm can be used, with
the path weight calculation replaced from the log likelihood of a
specific channel to a new metric designed for a compound set.
Moreover, additive structures are also suitable for belief
propagation algorithms. It is worth clarifying that the complexity
reduction discussed here rely on certain structured codes being
used, in the place of the random codes. In this paper, however, our
analysis will be based on the random coding argument, with the
implicit conjecture that there exists structured code resembling the
behavior of random codes under linear decoding. Mathematically, as
observed in \cite{cisnar,lapmis}, linear decoders are also
more interesting in that the geometric structure of decoders is
revealed, allowing the effects of ``mismatched" decoder to be
understood with engineering insights.

It is not surprising that for some compound channels, a linear universal decoder does not exist. In \cite{cisnar,lapmis}, it is shown that $S$ being convex and compact is a sufficient condition for the existence of linear universal decoders. In this paper, we give a more general sufficient condition for a set to admit a capacity achieving linear decoder, namely that $S$ is {\it one-sided}, following some geometric argument that will be made clear later. For more general compound sets, in order to achieve the capacity, we have to resort to a relaxed restriction of the decoders, which we call {\it generalized} linear decoders. A generalized linear decoder, for example, the well-known generalized loglikelihood ratio test (GLRT), maximizes a finite number, $K$, of decoding metrics, $d_1, d_2, \ldots, d_K$. The decoding map can then be written as
\begin{eqnarray*}
\arg\max_m \vee_{k=1}^K d^n_k(x_m, y) = \arg\max_m \vee_{k=1}^K \sum_{i=1}^n d_k( x_m(i), y(i)).
\end{eqnarray*}
Here, the receiver calculates in parallel $K$ additive metrics for each codeword, and decodes to the codeword with the highest among the total $2^{nR} \times K$ scores. In order such a generalized linear decoder to have a manageable complexity, we emphasize the restriction that $K$ has to be finite. In particular, it should not increase with the codeword length $n$.
For example the decoder proposed in \cite{BBT}, a mixture of likelihoods over all possible channels, in general might require averaging over polynomial($n$) channels. In addition, optimizing the mixture of additive metrics, i.e. $\arg\max_m \frac{1}{K} \sum_{k=1}^K d_k(x_m,y)$, cannot be solved by computing $K$ parallel additive metric optimizations: the codewords having the best scores for each of the $K$ metrics may not be the only candidates for the best score of the mixture
 of the metrics; on the other hand, if we consider a generalized linear decoder, the codewords having the best score for each of the $K$ metrics are the only one to be considered for the maximum of the $K$ metrics. 

The main result of this paper is the construction of generalized linear decoders that achieve compound channel capacity on most compound sets. As to be shown in Section \ref{sec:formulation}, this construction requires solving some rather complicated optimization problems involving the Kullback-Leibler (KL) divergence (like almost every other information theoretical problem). To obtain insights to this problem, we introduced in Section \ref{sec:vn} a special tool: local geometric analysis. In a nutshell, we focus on the special cases where the two distributions in the KL divergence are ``close" to each other, which can be thought in this context as approximating the given compound channels by very noisy channels. In this local setting, information theoretical quantities can be naturally understood as quantities in an inner product space, where conditional distributions and decoding metrics correspond to vectors; divergence and mutual information correspond to squared norms and the data rate with mismatched linear decoders can be understood with projections. The relation between these quantities can thus be understood intuitively. While the results from such local approximations only apply to the special very noisy cases, we show in Section \ref{sec:lift} that some of these results can be ``lifted" to the naturally corresponding statements about general cases. Using this approach, we derive the following main results of the paper. 
\begin{itemize}
\item First we derive a new condition on $S$ to be ``one-sided", cf. Definition \ref{onesideddef}, under which a linear decoder, which decodes using the log likelihood of the worst channel over the compound set, achieves capacity. This condition is more general than the previously known one, which requires $S$ to be convex;
\item Then, we show in our main result, that if the compound set $S$ can be written as a finite union of one sided sets, then a generalized linear decoder using the log a posteriori distribution of the worst channels of each one-sided subset achieves the compound capacity; in contrast, GLRT using these worst channels is not a universal decoder.
\end{itemize}

Besides the specific results on the compound channels, we also like to emphasize the use of the local geometric analysis. As most of multi-terminal information theory problems involve optimizations of K-L divergences, often between distributions with high dimensionality, we believe the localization method used in this paper can be a generic tool to simplify these problems. Focusing on certain special cases, this method is obviously useful in providing counterexamples to disprove conjectures. However, we also hope to convince the readers that the insights provided by the geometric analysis can be also valuable in solving the general problem. For example, our definition of one-sided sets and the use of log a posteriori distributions as decoding metrics can be seen as ``naturally" suggested by the local analysis.  

In the next section, we will start with the precise problem formulations and notations.

\section{Linearity and Universality}
\label{sec:formulation}

We consider discrete memoryless channels with input and output alphabets $\X$ and $\Y$, respectively.
The channel is often written as a probability transition matrix, $W$, of dimension $|\X | \times |\Y |$, each row of which denotes the conditional distribution of the  
output, conditioned on a specific value of the input. We are interested in the compound channel, where $W$ can be any elements of a given set $S$, referred to as the set of possible channels, or the {\it compound set}. For convenience, we assume $S$ to be compact. The value of the true channel is assumed to be fixed for the entire duration of communications, but not known to either the transmitter or the receiver; only the compound set $S$ is assumed to be known at both.

We assume that the transmitter and the receiver operates synchronously over blocks of $n$ symbols. In each block, a data message $m \in \{1, 2, \ldots, 2^{nR}\}$ is mapped by an encoder
\begin{eqnarray*}
F_n : \{1, 2, \ldots, 2^{nR}\} \mapsto \X^n
\end{eqnarray*}
to  $F_n(m)=x_m \in \X^n$, referred to as the $m^{th}$ codeword. The receiver observes the received word, drawn from the distribution
\begin{eqnarray*}
W^n(y|x_m) = \prod_{i=1}^n W(y(i)|x_m(i))
\end{eqnarray*}
and applies a decoding map
\begin{eqnarray*}
G_n: \Y^n \mapsto \{1, 2, \ldots, 2^{nR}\}.
\end{eqnarray*}

The average probability of error, averaged over a given code $(F_n, G_n)$, for a specific channel $W$, is written as
\begin{eqnarray*}
P_e(F_n, G_n, W) = \frac{1}{2^{nR}} \sum_{m=1}^{2^{nR}}  \sum_{\{y: G_n(y)\neq m\}} W^n(y|x_m).
\end{eqnarray*}

A rate $R$ is said to be achievable for the given compound set $S$ iff for any $\eps >0$, there exists a large enough block length $n$, and $(F_n, G_n)$ with rate at least $R$, such that for all $W\in S$, $P_e(F_n, G_n, W) < \eps$. The supremum of such achievable rates is called the {\it compound channel capacity}, written as $C(S)$. The following result from Blackwell et.al.  gives the compound channel capacity in general.

\begin{lemma} {\bf Compound Channel Capacity} \cite{BBT}
\begin{eqnarray}
\label{eqn:cc}
C(S) = \max_{P_X} \inf_{W\in S} I(P_X, W).
\end{eqnarray}
\end{lemma}
\vspace{.3cm}
\noindent{\it Remark:} The random coding argument is often used in proving the coding theorem for a fixed channel. By showing that the error probability, averaged over the ensemble of random codes, approaches $0$ as $n$ increases, one can draw the conclusion that there exists at least one sequence of codes, for which the probability of error, averaged over the specific codes, is driven to $0$. A similar argument is used in compound channels. Here, it is however not enough to show that the ensemble average error probability is small for every $W$. Since the ``good" codes for different channels can in principle be different, this is not enough to guarantee the existence of a single code that is universally good for all possible channels. The random coding argument is strengthened in \cite{BBT} to show that universally good code indeed exists. The approach used in \cite{BBT}, to show that the sets of good codes corresponding to every possible channel have non-empty intersection, has been used as a standard method to study compound channels. In this paper, we are focused on designing efficient decoders, which is interesting since the optimal maximum likelihood decoder is voided by the channel's law ignorance. We will not be particularly concerned about finding a good codebook, or even the optimal input distribution. To simplify our discussions, we will, for most of our results, only show that the ensemble average error probability can be made small, when decoders discussed in the paper are used. Arguments similar to that of \cite{BBT} can be used to show that the error probability can be made small when appropriately chosen codes are used.

Now before we proceed to define decoders, we need to define some notations:
\begin{itemize}
\item We always assume that we are working with the optimal input distribution $P_X$ for the considered compound set $S$, i.e.
$$P_X= \arg \max_{P } \inf_{W \in S} I(P,W)$$
(if the maximizers were not to be unique, we pick arbitrarily one of them).
Therefore, $\inf_{W \in S} I(P_X,W)$ is the compound channel capacity for a compound set $S$. However, the results in this paper can be stated for arbitrary input distributions (not necessarily optimal), the only difference would then be that we would talk about mutual informations instead of capacities.
\item For convenience, we assume that $S$ is compact. We define $W_S=\arg \min_{W \in S} I(P_X,W)$, and call it the {\it worst channel of} $S$ when the minimizer is unique; $I(P_X,W_S)$ is then the compound channel capacity for a compound set $S$.
We make the convention that each time a worst channel is considered throughout the paper for any set, the set in question is compact.
\item $W_0 \in S$ denotes the true channel;
\item For a joint distribution $\mu$ on $\X\times \Y$; $\mu_X$ and $\mu_Y$ denote respectively the $X$ and $Y$ marginal distributions; and $\mu^p =\mu_X \times \mu_Y$ the induced product distribution. Note that $\{\mu_X = P_X, \mu_Y= (\mu_0)_Y \} \Leftrightarrow \mu^p =\mu_0^p$
\item $\mu = P_X \circ W$ denotes the joint distribution with $P_X$ as the $X$ marginal distribution and $W$ as the conditional distribution. For example, the mutual information
\begin{eqnarray*}
I(P_X, W) = D( P_X \circ W \| (P_X \circ W)^p)
\end{eqnarray*}
where $D(\cdot \| \cdot)$ is the Kullback-Leibler divergence.
\end{itemize}

The decoders we consider has the following form. Upon receiving $y$, it computes, for each codeword $x_m$, a score $d^n(x_m, y)$, and decodes to the message corresponding to the highest score. Here, $d^n:  \X^n \times \Y^n\mapsto \mR$ is also called a {\it decoding metric}. Note the restriction here is that the score for codeword $x_m$ does not depend on other codewords. Such decoders are called $\alpha$-decoders in \cite{cisnar}. As an example, the maximum mutual information (MMI) decoder has a score defined as
\begin{eqnarray*}
d^n_{\sf MMI}(x_m, y) = I(\hat{P}_{(x_m, y)})
\end{eqnarray*}
where $\hat{P}$ denotes the empirical distribution. To be specific, $\forall a\in \X, b\in \Y$
\begin{eqnarray*}
\hat{P}_{(x_m, y)} (a,b) = \frac{1}{n} \left| \left\{i: (x_m(i), y(i)) = (a,b) \right\}\right|,
\end{eqnarray*}
and $I(\mu)$ denotes the mutual information, as a function of the joint distribution $\mu$ on $\X \times \Y$.

It is well known that the MMI decoder is universal; when used with the optimal code, it achieves the compound channel capacity on any compound sets. In fact, there are other advantages of the MMI decoder: it does not require the knowledge of $S$; and it achieves universally the random coding error exponent \cite{ciskon}. Despite these advantages, the practical difficulties to implement an MMI decoder prevents it from becoming a real ``universally used" decoder. As empirical distributions are used in computing the scores, it is difficult to store and update the scores, even when a structured codebook is used. The main goal of the current paper is to find linear decoders that can, like the MMI decoder, be capacity achieving on compound channels. 

\begin{definition} {\bf Linear Decoder}\\
We refer to a map
\begin{eqnarray*}
d: \X \times \Y \mapsto \mR
\end{eqnarray*}
as a single-letter metric. A linear decoder induced by $d$ is defined by the decoding mapping: 
\begin{eqnarray*}
&&G_n(y) = \arg\max_m d^n(x_m, y) \\
\mbox{ where} &&d^n(x_m, y) = \frac{1}{n}\sum_{i=1}^n d(x_m(i), y(i)) = E_{\hat{P}_{(x_m,y)}}[d]
\end{eqnarray*}
\end{definition}
Note that the reason why such decoders are called linear decoders ($d$-decoders in \cite{cisnar}) is to underline the fact that the decoding metric is additive, i.e. is a linear function of the empirical distribution $\hat{P}_{(x_m, y)}$. The decoding metric $d^n$ for any $n$ of a linear decoder is naturally defined by the single-letter metric $d$ through the additive structure.

The advantages of using linear decoders have been discussed thoroughly in \cite{cisnar,lapmis,laprev}, and also briefly in the introduction. In short, when used with structured codes, one can replace the log likelihood metric in a conventional decoder by a well designed single-letter metric. This way, with little changes in the decoder designs, one can have a decoder for the compound channel with much less complexity. 

Unfortunately, there are some examples for which no linear decoder can achieve the compound capacity. The
most well-known example is the compound set with two binary symmetric channels, with crossover probabilities of $1/4$ and $3/4$, respectively. To address the decoding challenge of these cases, we will need a slightly more general version of linear decoders.

\begin{definition} {\bf Generalized Linear Decoder}\\
Let $d_1, d_2, \ldots, d_K$ be $K$ single-letter metrics, where $K$ is a finite number. A generalized linear decoder induced by these metrics is defined by the decoding map:
\begin{eqnarray*}
G_n(y) &=& \arg\max_m \vee_{k=1}^K \sum_{i=1}^n d_k(x_m(i), y(i))\\
 &=&\arg\max_m \vee_{k=1}^K E_{\hat{P}_{(x_m,y)}}[d_k]
\end{eqnarray*}
\end{definition}
\vspace{.3cm}
Note that $\vee$ denotes the maximum, and it is crucial that $K$ is a finite number, which does not depend on the code length $n$.

As an example, the maximum likelihood decoder, of a given channel $W$, is a linear decoder induced by $$d_{\sf ML}(a,b) = \log W(b|a), \quad\forall a\in \X, b\in \Y.$$

It is well known that for a given channel $W$, the ML decoder, used with the random codes from the optimal input distribution, is capacity achieving. If the channel knowledge is imperfect, for example, the decoder uses ML rule for channel $W_1$ while the actual channel is $W_0$, the mismatch in the decoding metric causes the achievable data rate to decrease. This effect is studied in \cite{cisnar,lapmis}, the result is quoted in the following Lemma. For convenience, we also included a brief sketch of the proof.

\begin{lemma}\label{mism}\cite{cisnar,lapmis}
For a DMC $W_0$, using a random codebook with input distribution $P_X$, if the decoder is linear and induced by $d$, the following data rate can be achieved 
\begin{eqnarray}
R(P_X, W_0, d) = \inf_{\mu \in \A} D(\mu\| \mu_0^p)
\end{eqnarray}
where $\mu_0= P_X\circ W_0$, and $\mu_0^p$ is the product distribution with the same $X$ and $Y$ marginal distributions as $\mu_0$ and the optimization is over the following set of joint distributions on $\X \times \Y$,
\begin{eqnarray}
\A =  \{\mu: \mu_X = P_X, \mu_Y= (\mu_0)_Y, E_{\mu} [d] \geq E_{\mu_0}[d] \}. \label{misset}
\end{eqnarray}
\end{lemma}
\vspace{.2cm}
As discussed in \cite{lapmis}, this expression, even for the optimal $P_X$, does not give in general the highest achievable rate under the mismatched scenario. If the input alphabet is binary, it does so, otherwise it only gives the highest rate that can be achieved for codes that are drawn in a random ensemble.

\begin{proof}
This is a simple application of large deviations. By a typicality argument, the transmitted codeword, say, $x_1$, and the received word $y$ have joint empirical distribution close to $\mu_0$, and thus has a score
$$ d^n (x_1, y) > E_{\mu_0}[d] - \delta := \gamma$$
for an arbitrarily small $\delta>0 $ with a high probability when $n$ is large enough. Now an error occurs only if there is an incorrect codeword, whose score is above $\gamma$. For a particular codeword, $x_2$, this occurs with probability
\begin{eqnarray*}
P( d^n(x_2, y) > \gamma) \leq \exp\left[-n \left(\min_{\mu: E_\mu[d] > \gamma} D(\mu\|\mu_0^p)-\delta\right) \right],
\end{eqnarray*}
using the fact that $x_2$ is independent of $y$ with an i.i.d. $P_X$ distribution. The optimization is over the joint distributions $\mu$ with the correct $X$ and $Y$ marginal distributions. Now applying union bound, the probability
$$ P( \exists i \neq 1, \mbox{ s.t. } d^n(x_i, y) > \gamma) \leq  2^{nR} \cdot  P( d^n(x_2, y) > \gamma).$$ 
Moreover, the empirical distribution of $x_2, y$ is arbitrarily close to $\mu_0^p$ with probability one. 
Hence, if $R <  R(P_X, W_0, d)$ as defined in the lemma's statement, the above probability can be made arbitrarily small by taking $\delta$ small enough.
\end{proof}

With a similar proof as for previous result, the following lemma can also be proved.
\begin{lemma}
When the true channel is $W_0$ and a generalized linear decoder induced by the single-letter metrics $\{d_k\}_{k=1}^K$ is used, we can achieve the following rate
\begin{equation}
R(P_X, W_0, \{d_k\}_{k=1}^K) = \min_{\mu \in \A} D(\mu \| \mu_0^p) \label{defR}
\end{equation}
where
\begin{eqnarray*}
 \A &=& \{ \mu: \mu_X = P_X, \mu_Y = (\mu_0)_Y, \\
 && \qquad\vee_{k=1}^K E_\mu[d_k] > \vee_{k=1}^K E_{\mu_0}[d_k] \}
\end{eqnarray*}
\end{lemma}
\vspace{.3cm}
Note that $R(P_X, W_0, \{d_k\}_{k=1}^K)$ can equivalently be expressed as
\begin{equation}
R(P_X, W_0, \{d_k\}_{k=1}^K) = \min_{\mu \in \A_1} D(\mu \| \mu_0^p) \wedge\ldots \wedge \min_{\mu \in \A_K} D(\mu \| \mu_0^p) \label{projs}
\end{equation}
where
\begin{eqnarray*}
 \A_k &=& \{ \mu: \mu_X = P_X, \mu_Y = (\mu_0)_Y, \\
 && \qquad  E_\mu[d_k] > \vee_{j=1}^K E_{\mu_0}[d_j] \}, \quad \quad \forall 1 \leq k \leq K.
\end{eqnarray*}

Now we are ready for the main problem studied in this paper. For any given compound set $S$, let the compound channel capacity be $C(S)$ and the corresponding optimal input distribution be $P_X$. We would like to find $K$ and $d_1, \ldots, d_K$, such that
$$R(P_X, W_0, \{d_k\}_{k=1}^K) \geq C(S)$$ for every $W_0 \in S$.\\
%

If this holds, the generalized decoder induced by the metrics $\{d_k\}_{k=1}^K$ is capacity achieving on the compound set $S$ (i.e., using analogue arguments as for the achievability proof of the compound capacity in \cite{BBT}, there exists a code book that makes the overall coding scheme capacity achieving).


\section{The Local Geometric Analysis}
\label{sec:vn}
We know that the divergence is not a distance between two
distributions. However, if its two arguments are close enough, the
divergence is approximately a squared norm, namely for any probability distribution $p$ on $\Z$ (where $\Z$ is any alphabet) and for any $v$ s.t. $\sum_z
v(z)p(z)=0$, we have
\begin{eqnarray}
\label{eqn:local}
D(p(1 +\eps v) \| p) = \frac{1}{2}\eps^2 \sum_{z \in \Z} v^2(z) p(z)  + o(\eps^2).
\end{eqnarray}
This is the main tool used in this section. For convenience, we
define
\begin{eqnarray*}
\|v\|^2_p= \sum_{z\in \Z} v^2(z) p(z)
\end{eqnarray*}
which is the squared $l_2$-norm of $v$, with weight measure $p$.
Similarly, we can define the weighted inner product,
\begin{eqnarray*}
\langle u, v\rangle_p = \sum_{z\in \Z}  u(z) v(z) p(z)
\end{eqnarray*}
With these notations, one can write the approximation
(\ref{eqn:local}) as $$D(p(1+\eps v)\|p)=
\frac{\eps^2}{2} \|v\|^2_p + o(\eps^2)$$
Ignoring the higher order term, the above approximation can greatly
simplify many optimization problems involving K-L divergences. In information theoretic problems dealing with discrete channels, such
approximation is tight for some special cases such as when the
channel is very noisy.

\noindent
In general, very noisy channel means that the channel output weakly
depends on the input. If the conditional probability of observing
any output does not depend on the input (i.e. the transition
probability matrix has constant columns), we have a ``pure noise''
channel. So a very noisy channel should be somehow close to such a
pure noise channel. Formally, we consider the following family of
channels:
$$ W_{\eps}(b|a)= \Pn(b)(1+\eps L(a,b)), $$
where $L$ satisfies for any $a \in \X$
\begin{eqnarray}
 \sum_{b \in \Y} L(a,b)\Pn(b) = 0. \label{constrL}
\end{eqnarray}
\noindent
We say that $W_{\eps}$ is a very noisy channel if $\eps \ll
1$. In this case, the conditional distribution of the output,
conditioned on any input symbol, is close to a distribution $\Pn$ (on $\Y$),
which can be thought as the distribution of pure noise.
Each of these channels, $W_\eps(\cdot| \cdot)$, can
be viewed as a perturbation from a pure noise channel $\Pn$, along
the direction specified by $L(\cdot, \cdot)$.

This way of defining very noisy channel can be found in
\cite{telvn,gallager1}. In fact, there are many other possible ways to
describe a perturbation of distribution. For example, readers
familiar with \cite{ama} might feel it natural to perturb
distributions along exponential families. Since we are interested
only in small perturbations, it is not hard to verify that these
different definitions are indeed equivalent.

When an input distribution $P_X$ is chosen, the corresponding output
distribution, over the very noisy channel, can be written as,
$\forall b\in \Y$,
\begin{eqnarray*}
P_{Y,\eps}(b)&=& \sum_{a\in \X} P_X(a) W_\eps(b|a) \\
&=& \Pn(b) \left( 1+\eps \sum_a P_X(a)L(a, b) \right)\\
&=& \Pn(b) (1+ \eps \bar{L}(b))
\end{eqnarray*}
where $\bar{L}(b) = \sum_a P_X(a)L(a, b)$, $\forall a \in \X$.\\
Hence, a codeword which is sent and the received output have components which are i.i.d. from the following distribution
$$P_X \circ W_\eps = P_X  \Pn (1+\eps L ),$$
and similarly, the codeword which is not sent and the received output have components which are i.i.d. from the following distribution
$$ (P_X \circ W_\eps )^p=P_X  \Pn (1+\eps \bar{L}).$$
Therefore, the mutual information for very noisy channels is given by
\begin{eqnarray*}
I(P_X, W_\eps) &=& D(P_X  \Pn (1+\eps L ) \| P_X  \Pn (1+\eps \bar{L})) \\
&=& \frac{\eps^2}{2} \| \widetilde{L}\|^2  + o(\eps^2),
\end{eqnarray*}
where
\begin{eqnarray*}
\| \cdot \|=\| \cdot \| _{P_X \times \Pn}
\end{eqnarray*}
and
\begin{eqnarray*}
\widetilde{L}(a,b) \stackrel{\Delta}{=} L(a,b) - \bar{L}(b),
\end{eqnarray*}
 which we call the centered directions.

\subsection{Very Noisy with Mismatched Decoder}
As stated in Lemma \ref{mism}, for an input distribution $P_X$, a mismatched linear decoder induced by the metric $d$, when the true channel is $W_0$, can achieve the following rate
$$\inf_{\mu \in \A} D(\mu\|\mu_0^p)$$
where
$$\A=\{\mu : \,  \mu_X=P_X,\, \mu_Y=(\mu_0)_Y,  E_{\mu} [d] \geq  E_{\mu_0} [d] \}.$$
Now, if the channels are very noisy, this achievable rate can be expressed in the following simple form.
\begin{proposition}\label{vnmismprop}
Let $W_{0,\eps}=\Pn(1+ \eps L_0)$ and $d_{\eps}=\log
W_{1,\eps}$, where $W_{1,\eps}=\Pn(1+ \eps L_1)$. For a
given input distribution $P_X$, we can achieve the following rate
\begin{eqnarray*}
\lim_{\eps \rightarrow 0} \frac{2}{\eps^2}
R(P_X, W_{0,\eps}, d_{\eps})
=
 \begin{cases}
 \frac{\langle
\widetilde{L}_0,\widetilde{L}_1\rangle^2}{\|\widetilde{L}_1 \|^2}, 
&   \text{when} \,  \langle
\widetilde{L}_0,\widetilde{L}_1\rangle \geq 0 \\ 
0, & \text{otherwise}.
\end{cases} 
\end{eqnarray*}
\end{proposition}
\vspace{.3cm}
Note that it is w.l.o.g. to consider the single-letter metric to be the log of a channel, however, we do restrict all channels to be around a common $\Pn$ distribution.\\

Previous result says that the mismatched mutual information obtained
when decoding with the linear decoder induced by the mismatched
metric $\log W_{1,\eps}$, whereas the true channel is
$W_{0,\eps}$, is approximately the projections' squared norm of
the true channel centered direction $\widetilde{L}_0$ onto the mismatched
centered direction $\widetilde{L}_1$. This result gives an intuitive picture
of the mismatched mutual information, as expected, if the decoder is
matched, i.e. $\widetilde{L}_0=\widetilde{L}_1$, the projections'
squared norm is $\|\widetilde{L}_0\|^2$, which is the very noisy
mutual information of $\widetilde{L}_0$; and the more orthogonal
$\widetilde{L}_1$ is to $\widetilde{L}_0$, the more mismatched the
decoder is, with a lower achievable rate (eventually 0).

\begin{proof}
For each $\eps$, the minimizer $\mu_\eps$ can be expressed as
$$\mu_\eps = P_X \Pn (1+ \eps L) $$
where $L$ is a function on $\X \times \Y$, satisfying
\begin{eqnarray*}
\sum_{a\in \X, b\in \Y} P_X(a)  \Pn(b) L(a,b) = 0
\end{eqnarray*}
and the two marginal constraints, resp.
\begin{align}
\label{constr1}
&\left(\mu_\eps\right)_X = P_X \Longleftrightarrow  \notag \\
&\sum_{b\in \Y} \Pn(b) L(a,b)=0, \forall a\in \X\\
\label{constr2}
&\left(\mu_\eps\right)_Y = \left(\mu_0\right)_Y
\Longleftrightarrow \nonumber\\
&\sum_{a\in\X} P_X(a) L(a,b) = \sum_{a\in\X} P_X(a) L_0(a,b), \forall
b\in \Y
\end{align}
Now the constraint $E_{\mu}[\log W_{1,\eps}] \geq E_{\mu_0}[\log
W_{1,\eps}]$ can be written as
\begin{eqnarray*}
&& \sum_{a \in \X,b \in \Y} P_X(a)\Pn(b)(1+\eps L(a,b))\\
 &&\qquad \cdot[\log \Pn + \log(1+\eps L_1(x,y))] \\
 &\geq&  \sum_{a \in \X,b \in \Y} P_X(a)\Pn(b)(1+\eps L_0(a,b))\\
 &&\qquad \cdot[\log \Pn + \log(1+\eps L_1(x,y))].
\end{eqnarray*}
Using a first order Taylor expansion for the two $\log$ terms, and
the marginal constraint \eqref{constr2}, we have that previous constraint
is equivalent to
\begin{eqnarray}
\langle L,L_1 \rangle  \geq \langle L_0, L_1\rangle + o(1), \label{constr3}
\end{eqnarray}
where 
\begin{eqnarray}
\langle \cdot, \cdot \rangle = \langle \cdot, \cdot \rangle_{P_X \times \Pn} .
\end{eqnarray}

Finally, we can write the objective function as
\begin{eqnarray*}
D(\mu_{\eps}\|\mu_{0, \eps}^p) &=& D\left(P_X\Pn(1+
\eps L) \| P_X \Pn(1+ \eps \bar{L}_0)\right)\\
&=& \frac{\eps^2}{2} \left\| L- \bar{L}_0 \right\|^2_{P_X\times
\Pn} + o(\eps^2)
\end{eqnarray*}

So we have transformed the original optimization problem into the
very noisy setting
\begin{eqnarray}
\lim_{\eps\to 0} \frac{2}{\eps^2} \inf_{\mu\in \A} D( \mu\|
\mu_{0, \eps}^p) =\inf_{L: \langle L, L_1\rangle \geq \langle
L_0, L_1\rangle} \lno L-\bar{L}_0 \rno^2 \label{optim}
\end{eqnarray}
where the optimization on the RHS is over $L$ satisfying the
marginal constraints \eqref{constr1} and \eqref{constr2}. 

Now this optimization can be further simplified. By noticing that
\eqref{constr2} implies $\bar{L}= \bar{L}_0$, we have that $L- \bar{L}_0
= L - \bar{L}$, which we defined to be $\widetilde{L}$. So $\widetilde{L}$ satisfies both marginal constraints and the constraint in \eqref{optim} becomes
\begin{eqnarray*}
\langle L, L_1\rangle \geq \langle
L_0, L_1\rangle &\Leftrightarrow& \langle \widetilde{L}, L_1 \rangle
\geq \langle L_0, L_1\rangle - \langle \bar{L}_0, \bar{L}_
1\rangle \\
&\Leftrightarrow& \langle\widetilde{L}, \widetilde{L}_1\rangle \geq
\langle\widetilde{L}_0, \widetilde{L}_1\rangle
\end{eqnarray*}
That is, both the objective and the constraint functions are now
written in terms of centered directions, $\widetilde{L}$. Hence, \eqref{optim} becomes
$$ \inf_{\wtil{L}: \langle \wtil{L}, \wtil{L}_1\rangle \geq \langle
\wtil{L}_0, \wtil{L}_1\rangle} \| \wtil{L} \|^2 $$
and we can simply recognize that, if $\langle \widetilde{L}_0,\widetilde{L}_1\rangle \geq 0$, the minimizer of this expression is obtained by the projection of $\wtil{L}_0$ onto $\wtil{L}_1$, with a minimum given by the projections' squared norm:
$$ \frac{\langle \widetilde{L}_0,\widetilde{L}_1\rangle^2}{\|\widetilde{L}_1\|^2},
$$
otherwise, if $\langle \widetilde{L}_0,\widetilde{L}_1\rangle < 0$, the minimizer is $\wtil{L}=0$, leading to a zero rate. 
\end{proof}

\noindent
{\it Remark:} We have just seen two examples where in the very noisy limit, information theoretic quantities have a natural geometric meaning, in the previously described inner product space. The cases treated in this section are the ones relevant for the paper's problem, however, following similar expansions, other information theoretic problems, in particular multi-user ones (e.g. broadcast or interference channels) can also be treated in this geometrical setting.
To simplify the notation, since the very noisy expressions scale with $\eps^2$ and have a factor $\frac{1}{2}$ in the limit, we denote by $\vnl$ the following operator:
$$T(\eps) \vnl  \lim_{\eps \searrow 0} \frac{2}{\eps^2} T(\eps).$$
We use the abbreviation VN for very noisy. Note that the main reason
why we use the VN limit in this paper is similar somehow to the reason why
we consider infinite block length in information theory: it gives us
a simpler model to analyze and helps us understanding the more
complex (not necessarily very noisy) general model. This makes the
VN limit more than just an approximation
for a specific regime of interest, it makes it an analysis tool of our problems, by setting them in a geometric framework where notion of distance and angles are this time well defined. Moreover, as we will show in section
\ref{lifting}, in some cases, results proven in the VN limit can
in fact be ``lifted" to results proven in the general cases.

\section{Linear Decoding for Compound Channel:\\ the Very Noisy Case}\label{ludvn}
In this section, we will study a special case of the compound
channel, the very noisy case. The local geometric analysis
introduced in the previous section can be immediately applied to
such problems. Throughout this process, we will develop a few
important concepts that will be used in solving the general compound
channel problems, in section \ref{lifting}. In the following, we
first make clear of our assumptions, and introduce some notations.
\begin{itemize}
\item
All the channels are very noisy, with the same pure noise
distribution. That is, all considered channels are of the form
\begin{eqnarray*}
W_{\eps} (b|a) = \Pn (b) (1+ \eps L(a,b)), \quad
\forall a\in \X, b\in \Y
\end{eqnarray*}
where $L$ satisfies $\sum_b \Pn(b) L(a,b)=0, \forall a$.
The compound set is hence depending on $\eps$, and is expressed as
$S_\eps = \{ \Pn(1+\eps L) | L \in \Sc  \}$, where $\Sc$
is the set of all possible directions. Hence, $\Sc$ together with the pure noise distribution $\Pn$, completely determine the compound set for any $\eps$. We refer to $\Sc$ as the compound set in the VN setting.
Note that $\Sc$ being convex, resp. compact, is the sufficient and
necessary condition that $S_\eps$ is convex, resp. compact, for
all $\eps$.
\item $P_X$ is fixed (it is the optimal input distribution) and we write $$\mu_{\eps} = P_X \Pn (
1+ \eps L), L \in \Sc$$ as the joint distribution of the input and
output over a particular channel.
For a given channel $W_{
\eps}$, the output distribution is $\Pn(1+ \eps \bar{L})$,
where
\begin{eqnarray*}
\bar{L}(b) = \sum_{a\in \X} L(a, b) P_X(a), \quad \forall b\in
\Y
\end{eqnarray*}
and as before, $\widetilde{L}= L- \bar{L}$.
We then denote $\widetilde{\Sc}= \{\widetilde{L}: L\in \Sc\}$. Again, the
convexity and compactness of $\Sc$ is equivalent to those of $\widetilde{\Sc}$. The only difference is that $\Sc$ depends on the channels only, whereas $\widetilde{\Sc}$ depends on the input
distribution as well. As we fix $P_X$ in this section, we use the conditions
$L\in \Sc$ and $\widetilde{L} \in \widetilde{\Sc}$ exchangeably.
\item As a convention, we often give an index, $j$, to the possible
channels, and we naturally associate the channel index (the joint distribution index) and the direction index, i.e. $W_{j, \eps} = \Pn ( 1+
\eps L_j)$ and $\mu_{j,\eps} = P_X \Pn (
1+ \eps L_j)$. In particular, we reserve $W_{0, \eps} = \Pn ( 1+
\eps L_0)$ for the true channel and use other indices, $L_1,
L_2,$ etc. for other specific channels.

\item If one considers the metrics to be the $\log$ of some channels, i.e., $d_j=\log W_{j, \eps}$,
$$ d_{j,\eps} = \log W_{j,\eps} =\log(\Pn) + \log(1+ \eps L_j) .$$
In general, the single-letter
decoding metric $d$ does not have to be the log likelihood of a
channel; and even if it is, the channel $W_{j,
\eps}$ does not have to be in the compound set.
\item We write all inner products and norms as weighted by $P_X\times
\Pn$, and omit the subscript:
$$ \langle \cdot, \cdot \rangle = \langle \cdot, \cdot \rangle_{P_X \times \Pn}. $$
\item Finally, 
$$ \min_{W \in S_\eps} I(P_X,W) = \frac{\eps^2}{2} \min_{L \in \Sc} \|\wtil{L} \|^2 + o(\eps^2)$$
and we define
$$L_\Sc= \arg \min_{L \in \Sc} \|\wtil{L} \|^2 ,$$
to be the worst direction and $ \lVert \wtil{L}_\Sc \rVert^2 $ is referred to as the very noisy compound channel capacity (on $\Sc$). 
\end{itemize}
We conclude this section with the following lemma, which will be frequently used in the subsequent.
\begin{lemma}\label{theid}
Let $L_i, L_j, L_k$ and $L_l$ be four directions and assume that $\sum_a P_X (a) L_i (a) = \sum_a P_X (a) L_k (a)$. We then have
\begin{align*}
&E_{\mu_{i,\eps}} \log W_{j,\eps} >  E_{\mu_{k,\eps}} \log W_{l,\eps} \\
& \vnl \langle L_i,L_j  \rangle - \frac{1}{2} \| L_j \|^2 > \langle L_k,L_l  \rangle - \frac{1}{2} \| L_l \|^2.
\end{align*}
\end{lemma}
\begin{proof}
Using a second order Taylor expansion for $\log (1 +\eps L_j)$, we have
\begin{align}
&E_{\mu_{i,\eps}} \log W_{j,\eps} = \sum P_X \Pn (1+\eps L_i) \log (\Pn (1+\eps L_j) ) \notag \\
&= \sum P_X \Pn \log \Pn \notag \\&+ \eps \sum P_X \Pn L_i \log \Pn
+ \eps \sum P_X \Pn L_j \notag \\
&+ \eps^2 \sum P_X \Pn L_i L_j - \eps^2 \frac{1}{2} \sum P_X \Pn L_j^2 \label{survive}
\end{align}
The only term which is zero in previous summation is the third term, namely $ \sum P_X \Pn L_j=0$, which is a consequence of the fact that $L_j$ is a direction (i.e. $ \sum \Pn L_j =0$). Now, when we look at the inequality $E_{\mu_{i,\eps}} \log W_{j,\eps} >  E_{\mu_{k,\eps}} \log W_{l,\eps}$, we can surely simplify the term $\sum P_X \Pn \log \Pn$, since it appears both on the left and right hand side. Moreover, using the assumption that $\sum_a P_X (a) L_i (a) = \sum_a P_X (a) L_k (a)$, we have $\sum P_X \Pn L_i \log \Pn = \sum P_X \Pn L_k \log \Pn $. Hence the only terms that survive in \eqref{survive}, when computing $E_{\mu_{i,\eps}} \log W_{j,\eps} >  E_{\mu_{k,\eps}} \log W_{l,\eps}$, are the terms in $\eps^2$, which proves the lemma.
\end{proof}

\subsection{One-sided Sets}\label{onesidedsec}
We consider for now the use of linear decoder (i.e., induced by only one metric).
We recall that, as proved in previous section, for
$W_{0,\eps}=\Pn(1+ \eps L_0)$ and $d_{\eps}=\log
W_{1,\eps}$, where $W_{1,\eps}=\Pn(1+ \eps L_1)$, we
have
\begin{eqnarray*}
\lim_{\eps \rightarrow 0} \frac{2}{\eps^2}
R(P_X, W_{0,\eps}, d_{\eps})
=
 \begin{cases}
\frac{\langle
\widetilde{L}_0,\widetilde{L}_1\rangle^2}{\|\widetilde{L}_1 \|^2}, 
&   \text{when} \,  \langle
\widetilde{L}_0,\widetilde{L}_1\rangle \geq 0 \\ 
0, & \text{otherwise}.
\end{cases} 
%
\end{eqnarray*}
This picture of the mismatched mutual information directly suggests
a first result.
Assume $\Sc$, hence $\wtil{\Sc}$, to be convex. 
By using the worse channel to be the only decoding metric, it is
then clear that the VN compound capacity can be achieved. In fact,
no matter what the true channel $\widetilde{L}_0 \in \wtil{\Sc}$ is,
the mismatched mutual information given by the projections' squared
norm of $\widetilde{L}_0$ onto $\widetilde{L}_ {\Sc}$ cannot be shorter than $\|\widetilde{L}_ {\Sc}\|^2$, which is the very noisy compound capacity of $\Sc$ (cf. Figure \ref{onesidedB}). 
This agrees with a result proved in \cite{cisnar}.

However, with this picture we understand that the notion of
convexity is not necessary. As long as the compound set is such that
its projection in the direction of the minimal vector stays on one
side, i.e., if the compound set is entirely contained in the half
space delimited by the normal plan to the minimal vector, i.e., if for any $L_0 \in \Sc$, we have $\langle \widetilde{L}_0,\widetilde{L}_ {\Sc}\rangle \geq 0$ and :
$$  \frac{\langle \widetilde{L}_0,\widetilde{L}_ {\Sc}\rangle^2}{\|\widetilde{L}_ {\Sc}\|^2} \geq \|\widetilde{L}_ {\Sc}\|^2,$$
we will achieve compound capacity by using the linear decoder induced by the
worst channel metric (cf. figure \ref{onesidedB} where $S$ is not
convex but still verifies the above conditions). We call such sets
one-sided sets, as defined in the following.
\begin{definition}\label{thedef} {\bf VN One-sided Set}\\
A VN compound set $\Sc$ is one-sided iff for any $L_0 \in
\Sc$, we have
\begin{align}
& \langle \widetilde{L}_0,\widetilde{L}_ {\Sc}\rangle \geq 0,\label{defonesided2} \\ 
& \frac{\langle \widetilde{L}_0,\widetilde{L}_ {\Sc}\rangle^2}{\|\widetilde{L}_ {\Sc}\|^2} \geq \|\widetilde{L}_ {\Sc}\|^2. \label{defonesided22}
\end{align}
Equivalently, a VN compound set $\Sc$ is one-sided iff for any $L_0 \in
\Sc$, we have 
\begin{eqnarray}
\| \wtil{L}_0 \|^2 - \| \wtil{L}_S \|^2 - \| \wtil{L}_0 - \wtil{L}_S \|^2 \geq 0.  \label{defonesided}
\end{eqnarray}
\end{definition}
\begin{figure}
\begin{center}
\includegraphics[scale=.64]{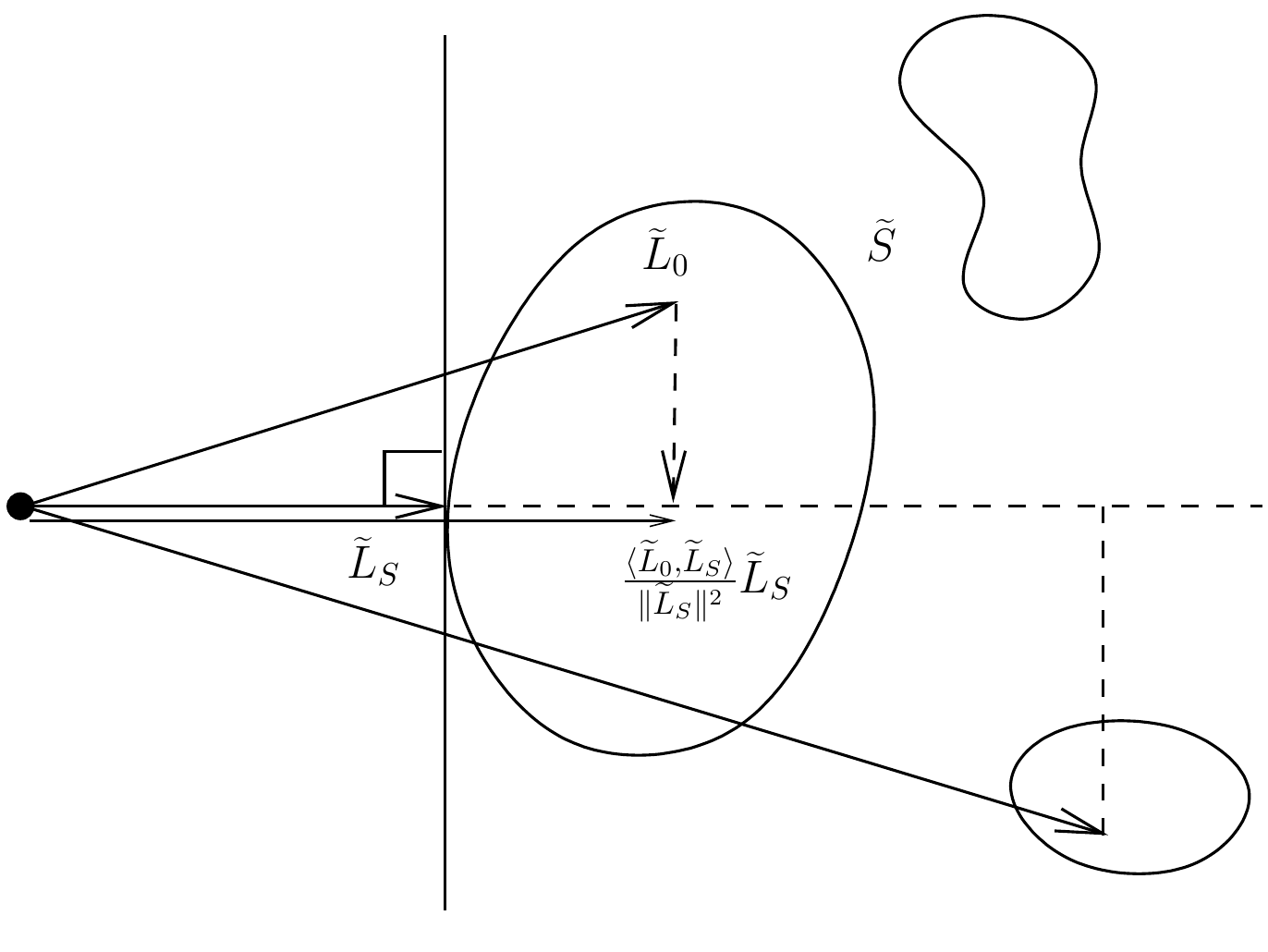} 
\caption{Very noisy one-sided compound set: in this figure, $\wtil{S}$ is the union of three sets.  The linear decoder induced by the worst channel metric $\log L_\Sc$ when the true channel is $L_0$ affords reliable communication for rates as large as the squared norm of the projection of $\wtil{L}_0$ onto $\wtil{L}_\Sc$. From the one-sided shape of the compound set, this projections' squared norm is always as large as the compound capacity given by the squared norm of $\wtil{L}_\Sc$.
}
\label{onesidedB}
\end{center}
\end{figure}
\vspace{.3cm}
\begin{proposition}
In the VN setting, the linear decoder induced by the worst channel metric $\log L_\Sc$ is capacity achieving for one-sided sets.
\end{proposition}
\vspace{.2cm}
The very noisy picture also suggests that the one-sided property is indeed necessary in order to be able to achieve the compound capacity with a single linear decoder. However, our main goal here is not motivated by results of this kind and we will not discuss this in more details. We now investigate whether we can still achieve compound capacity on non one-sided compound sets, by using generalized linear decoders.

\subsection{Finite Sets}
Let us consider a simple case of non one-sided set, namely when $S$ contains only two channels that are not satisfying the one-sided property in \eqref{defonesided}. We denote the set by
$$S=\{W_0,W_1\}.$$
and it contains the true channel $W_0$ and an arbitrary other channel $W_1$.
A first idea is to use a generalized decoder induced by the two metrics $d_1 = \log W_0$ and $d_2 = \log W_1$, i.e. decoding with the GLRT test using both channels, which defines the following decoding map
$$ \arg\max_{x_m} W_0^n(y|x_m) \vee W_1^n(y|x_m) .$$
The maximization of $W_0^n(y|x_m)$ corresponds to the maximization of an optimal ML decoder with the true channel, whereas the maximization of $W_1^n(y|x_m)$ corresponds to the maximization of ML decoder with a mismatched metric, which may have nothing to do with
the true channel metric. So we need to estimate how probable it is that a codeword which has not been sent appears highly plausible under the mismatched metric (i.e., an error event).
Using, \eqref{projs}, we can achieve the following rate with such a decoder:
\begin{align}
& R_0 \wedge R_1, \\
&\text{where } R_k =  \min_{\mu \in \A_k} D(\mu \| \mu_0^p), \,\,\, k=0,1
\end{align}
and
\begin{eqnarray*}
 \A_k &=& \{ \mu: \mu_X = P_X, \mu_Y = (\mu_0)_Y, \\
 && \qquad  E_\mu \log W_k > \vee_{j=0}^1 E_{\mu_0} \log W_j \}, \quad \quad \forall k=0,1.
\end{eqnarray*}
Note that $ \vee_{j=0}^1 E_{\mu_0} \log W_j = E_{\mu_0} \log W_0$, hence the expression of $\A_k$ simplifies to
\begin{eqnarray*}
 \A_k &=& \{ \mu: \mu_X = P_X, \mu_Y = (\mu_0)_Y, \\
 && \qquad  E_\mu \log W_k >  E_{\mu_0} \log W_0 \}, \quad \quad \forall k=0,1.
\end{eqnarray*}
Moreover, the compound capacity of $S$ is given here by $$C(S)=I(P_X,W_0) \wedge I(P_X,W_1).$$
We know that $R_0$ is the mutual information of $W_0$, i.e. $R_0 = I(P_X, W_0)$ (since it is the rate achieved with a ML decoder with a metric matched to the channel, as explained previously). So the generalized decoder that we are considering achieves compound capacity if $R_1 \geq C(S)$.
We check this here in the very noisy setting.
We use the notations and conventions defined previously for the VN setting,
and to compute the VN limit of $R_{1,\eps}$, we need the VN limits of $D(\mu_\eps \| \mu_{0,\eps}^p)$ and $\A_{1,\eps}$. We have
\begin{eqnarray*}
 D(\mu_\eps \| \mu_{0,\eps}^p) \vnl \|L -\bar{L}_0\|^2.  
\end{eqnarray*}
Moreover $ (\mu_\eps)_X = P_X$ for any $\eps$, since we assume that $L$ satisfies $\sum_b L(a,b) \Pn (b) = 0$ and
\begin{align*}
(\mu_\eps)_Y = (\mu_{0,\eps})_Y \vnl \bar{L}=\bar{L}_0.
\end{align*}
Finally, using lemma \ref{theid}, we have
\begin{align*}
&E_{\mu_\eps} \log W_{1,\eps} >  E_{\mu_{0,\eps}} \log W_{0,\eps} \\
& \vnl \langle L,L_1  \rangle - \frac{1}{2} \| L_1 \|^2 > \frac{1}{2} \| L_0 \|^2.
\end{align*}
Hence
\begin{align}
& \A_{1,\eps} \vnl
\{ L: \, \bar{L}=\bar{L}_0 , \langle L,L_1 \rangle  >  \frac{1}{2}(\| L_0 \|^2+\| L_1 \|^2) \} \notag \\
&= \{ \wtil{L}: \,  \langle \wtil{L},\wtil{L}_1 \rangle  >  \frac{1}{2}(\| L_0 \|^2+\| L_1 \|^2)- \langle \bar{L}_0,\bar{L}_1  \rangle \}. \label{puttil}
\end{align}
Note that we used $ \bar{L}=\bar{L}_0 $ to get \eqref{puttil} from its previous line.
Putting pieces together we get
\begin{align*}
& R_{1,\eps} \vnl \min_{L: \, \bar{L}=\bar{L}_0 , \langle \wtil{L},\wtil{L}_1 \rangle  >   \frac{1}{2}(\| L_0 \|^2+\| L_1 \|^2)- \langle \bar{L}_0,\bar{L}_1  \rangle}  \|L -\bar{L}_0\|^2 \\
& \qquad \quad \,\, =  \min_{L: \, \langle \wtil{L},\wtil{L}_1 \rangle  >   \frac{1}{2}(\| L_0 \|^2+\| L_1 \|^2)- \langle \bar{L}_0,\bar{L}_1  \rangle}  \| \wtil{L}\|^2 .
\end{align*}
We now are able to resolve the above minimization, and we get
\begin{align*}
 R_{1,\eps} \vnl    \frac{\left[\frac{1}{2}(\| L_0 \|^2+\| L_1 \|^2) - \langle \bar{L}_0,\bar{L}_1  \rangle \right]^2}{\| \wtil{L}_1\|^2}.
\end{align*}
Also,
$$ C(S_\eps) \vnl \|\widetilde{L}_0\|^2 \wedge \|\widetilde{L}_1\|^2.$$
Therefore, the inequality which allows us to verify locally if the proposed decoding rule achieves compound capacity, i.e. if $R_1 \geq C(S)$ in the VN setting, is given by
\begin{align}
 & R_{1,\eps} \geq C(S_\eps) \vnl \notag \\
 &  \frac{\left[\frac{1}{2}(\| L_0 \|^2+\| L_1 \|^2) - \langle \bar{L}_0,\bar{L}_1  \rangle \right]^2}{\| \wtil{L}_1\|^2} \geq \|\widetilde{L}_0\|^2 \wedge \|\widetilde{L}_1\|^2 . \label{tocheck}
\end{align}
But
\begin{align*}
 &\frac{1}{2}(\|L_0\|^2+\|L_1\|^2) - \langle \bar{L}_0,\bar{L}_1 \rangle \\
 &= \frac{1}{2}(\|\widetilde{L}_0\|^2+\|\widetilde{L}_1\|^2+\|\bar{L}_0-\bar{L}_1\|^2),
\end{align*}
hence, \eqref{tocheck} is equivalent to
$$\frac{1}{2}(\|\widetilde{L}_0\|^2+\|\widetilde{L}_1\|^2+\|\bar{L}_0-\bar{L}_1\|^2) \geq \|\widetilde{L}_0\|\|\widetilde{L}_1\| \wedge \|\widetilde{L}_1\|^2,$$
which clearly holds no matter what $L_0$ and $L_1$ are. \\
This can be directly generalized to any finite sets and we have the following result.
\begin{proposition}
In the VN setting, GLRT with all channels in the set is capacity achieving for finite compound sets, and generalized linear.
\end{proposition}

\subsection{Finite Union of One-sided Sets}\label{uniononesidedsec}
\subsubsection{Using ML Metrics}\label{ml}
In the previous sections, we have found linear, or generalized
linear, decoders that are capacity achieving for one-sided sets and
for finite sets. Next we consider compound sets that are finite
unions of one-sided sets and hope to combine our results in these
two cases. Assume
$$S=S_1 \cup S_2 , $$
where $S_1$ and $S_2$ are one-sided: in this section we consider
only the VN setting, hence saying that $S_1$ is one sided really
means that the VN compound set $\Sc_1$ corresponding to $S_{1,\eps}$
is one-sided according to Definition \ref{thedef}.

For a fixed input distribution $P_X$, let $W_1=W_{S_1}$ and
$W_2=W_{S_2}$ be the worst channel of $S_1, S_2$, respectively.(cf.
figure
\ref{2oc}).
\begin{figure}
\begin{center}
\includegraphics[scale=.95]{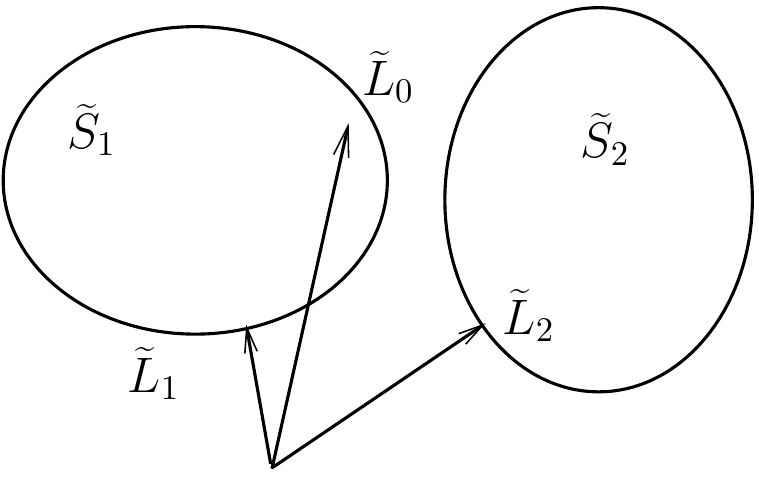}
\caption{A VN Compound set which is the union of two one-sided components, $\wtil{S}_1$ and $\wtil{S}_2$, drawn in the space of centered directions (tilde vectors)}
\label{2oc}
\end{center}
\end{figure}
A plausible candidate for a generalized linear universal decoder the
GLRT with metrics $d_1=\log W_1$ and $d_2=\log W_2$, hoping that a
combination of earlier results for finite and one-sided sets would
make this decoder capacity achieving. Say w.l.o.g. that $W_0 \in
S_1$. Using
\eqref{projs}, the following rate can be achieved with the proposed
decoding rule:
$$R(P_X, W_0, \{d_k\}_{k=1}^K)=R_1 \wedge R_2$$
where
$$R_k = \inf_{\A_k} D(\mu\|\mu_0^p) , \quad k=1,2$$
and for $k=1,2$,
\begin{equation}
\label{eqn:Ak}
\A_k =\{\mu : \,  \mu^p=\mu_0^p, E_{\mu} \log W_k
\geq   \vee_{l=1}^2E_{\mu_0} \log W_l \},
\end{equation}
Note that we are using similar notations for this section as for the previous
one, although the sets $\A_k$ and rates $R_k$ are now given by
different expressions. We also use $\mu^p=\mu_0^p$ to express in a
more compact way that the marginals of $\mu$ and $\mu_0$ are the
 same.

Since $W_1$ and $W_2$ are the worst channel for $P_X$ in each component, the compound capacity over $S=S_1 \cup S_2$ is
$$C(S) = I(P_X,W_1) \wedge I(P_X,W_2).$$
In the finite compound set case of previous section, we further
simplified the expression of the $\A_k$'s, since we the maximum in
$\vee_{l=1}^2 E_{\mu_0} \log W_l $ could be identified. This is no
longer the case here, and we have to consider both cases, i.e.:
\begin{eqnarray}
\label{eqn:case1}
&\text{Case 1:} \,\,\, E_{\mu_0} \log W_1 \geq E_{\mu_0} \log W_2\\
\label{eqn:case2}
&\text{Case 2:} \,\,\, E_{\mu_0} \log W_1 \leq E_{\mu_0} \log W_2.
\end{eqnarray}
In order to verify that the decoder is capacity achieving, we need
to check if both $R_1$ and $R_2$  are greater
than or equal to the compound capacity $C(S)$, no matter which of case 1 or case 2 occurs. Thus, there are
totally $4$ inequalities to check. While checking these cases is
somewhat tedious, we will, in the following, go through each of them
carefully and point out a specific case that is problematic, before
giving a counterexample where GLRT with the worst channels is in fact {\it
not} capacity achieving. Later when we propose a capacity achieving
decoder, we will go through a similar procedure in a more concise
way.

Note that under case 1,
$$\text{For case 1:}\,\,\, \A_1 =\{ \mu : \,  \mu^p=\mu_0^p, E_{\mu} \log W_1 \geq  E_{\mu_0} \log W_1 \},$$
which has the form of the constraint set for $R(P_X,W_0,d_1)$ expressed in \eqref{misset}.
Hence we have
\begin{align}
&\text{For case 1:} \,\,\,  R_1 = R(P_X,W_0,\log W_1).
\end{align}
As shown in section \ref{onesidedsec}, $R(P_X,W_0,\log W_1)$ becomes in the VN limit:
\begin{align}
R(P_X,W_{0,\eps},\log W_{1,\eps}) \vnl \frac{\langle \widetilde{L}_0,\widetilde{L}_ {1}\rangle^2}{\|\widetilde{L}_ {1}\|^2}
\end{align}
(note that since $\Sc_1$ is one-sided, $\langle \widetilde{L}_0,\widetilde{L}_ {1}\rangle \geq 0$).
Also, in the VN limit, $C(S_\eps)$ becomes $\| \wtil{L}_1 \|^2 \wedge \| \wtil{L}_2 \|^2$, hence 
\begin{align}
\text{For case 1:} \,\,\, R_{1,\eps} \stackrel{?}{\geq} C(S_\eps) \vnl \frac{\langle \widetilde{L}_0,\widetilde{L}_ {1}\rangle^2}{\|\widetilde{L}_ {1}\|^2 } \stackrel{?}{\geq} \| \wtil{L}_1 \|^2 \wedge \| \wtil{L}_2 \|^2. \label{r1}
\end{align}
But we assumed that $\Sc_1$ is one-sided and that $L_1$ is the worst
direction of $\Sc_1$. Moreover, we assumed that $W_0 \in S_1$, i.e.
$L_0 \in \Sc_1$. Hence, \eqref{r1} holds by definition of one-sided
sets, cf. def. \ref{defonesided} (with this definition, \eqref{r1} holds with $\| \wtil{L}_1 \|^2$ on the right hand side, hence it holds for $\| \wtil{L}_1 \|^2 \wedge \| \wtil{L}_2 \|^2$).

For case 2, i.e. when $E_{\mu_0} \log W_1 \leq E_{\mu_0} \log W_2$, we
have $R_1 = \inf_{\mu\in\A_1} D(\mu\|\mu_0^p)$, where this time $\A_1$ is given by
\begin{eqnarray}
\text{For case 2:}\,\,\, \A_1 = \{ \mu:
\mu^p=\mu_0^p, E_\mu \log W_1 \geq E_{\mu_0} \log
W_2 \}
\end{eqnarray}
Note that, by definition of case 2, the constraint set $\A_1$ is smaller than the constraint set $\mathcal{B}$ given below: 
\begin{align}
&\A_1 = \{ \mu:
\mu^p=\mu_0^p, E_\mu \log W_1 \geq E_{\mu_0} \log
W_2 \}\nonumber\\
&\subset \mathcal{B}=\{ \mu: \mu^p=\mu_0^p, E_\mu \log W_1 \geq E_{\mu_0} \log
W_1 \}
\label{eqn:R1case2}
\end{align}
hence, 
$$\inf_{\mu \in \A_1} D(\mu \| \mu_0^p) \geq \inf_{\mu \in \mathcal{B}} D(\mu \| \mu_0^p).$$
But $\mathcal{B}$ is the constraint set appearing in $R(P_X, W_0, \log W_1)$, which means that
$$\inf_{\mu \in \mathcal{B}} D(\mu \| \mu_0^p) = R(P_X, W_0, \log W_1),$$
therefore, under case 2, we showed that $R_1 \geq R(P_X, W_0, \log W_1)$. Now, as shown before,  $R(P_X, W_0, \log W_1)$ is locally lower bounded by $I(P_X, W_1)\geq C(S)$, by the one-sided assumption on $\Sc_1$. \\

Hence, we have just shown that $R_1 \geq C(S)$, both under case 1 and 2. \\

Next, we check whether $R_2 = \inf_{\mu\in\A_2} D(\mu\|\mu_0^p) \geq
C(S)$ holds or not.  We have again to check this for case 1 and 2. This time we start with case 2.
Note that the expression of $R_2$ in case 2 is perfectly symmetric to the expression of $R_1$ in case 1, we just have to swap the indices 1 and 2, hence
\begin{align*}
&\text{For case 2:} \,\,\,  R_2 = R(P_X,W_0,\log W_2).
\end{align*}
and the inequality we need to check in the very noisy case is
\begin{align}
\text{For case 2:} \,\,\, R_{2,\eps} \stackrel{?}{\geq} C(S_\eps) \vnl \frac{\langle
\widetilde{L}_0,\widetilde{L}_2\rangle^2}{\|\widetilde{L}_2\|^2} \stackrel{?}{\geq}
\|\widetilde{L}_1\|^2 \wedge \|\widetilde{L}_2\|^2. \label{bug}
\end{align}
However, the one-sided property does not apply anymore, since we
assumed that $L_0 $ belongs to $\Sc_1$ and not $\Sc_2$. Indeed, if we have no restriction on the positions of $\wtil{L}_0$ and $\wtil{L}_2$, \eqref{bug} can be zero.
Comparing this with the case of a single one-sided set, we see this is exactly the difficulty of analyzing generalized linear decoders.
Using multiple metrics, especially $d_2=\log W_2$, which does not
have any one-sided relation with the actual channel $W_0$, causes an
extra chance of making errors: an incorrect codeword can appear very
plausible according to metric $d_2$. The probability for this to
happen is captured by the rate $R_2$. On the other hand, there is
also a lower target: \eqref{bug} should not hold for any possible 
$\wtil{L}_0$, $\wtil{L}_1$ and $\wtil{L}_2$, \eqref{bug} should hold when these centered directions are satisfying case 2. Moreover, the compound capacity is now the minimum between the mutual informations $\|\widetilde{L}_1\|^2$ and $\|\widetilde{L}_2\|^2$. 
One might hope that the combination
of all these effects leads to $R_2> C(S)$ and hence a capacity
achieving decoder design. Unfortunately, this is not the case.

\vspace{0.5cm}

\begin{proposition}
\label{prop:badglrt}
In the VN setting and for compound sets having a finite number of
one-sided components, GLRT with the worst channel of each component
is not capacity achieving.
\end{proposition}
\vspace{.3cm}
{\it Counterexample:} Let $\X=\Y=\{0,1\}$, $P_X=\Pn=\{1/2,1/2\}$,
$$L_0=\begin{pmatrix}
     -2 & 2   \\
      -7 &  7
\end{pmatrix},L_1=\begin{pmatrix}
     2 & -2   \\
      0 &  0
\end{pmatrix} \text{ and }L_2= \begin{pmatrix}
     -1 & 1   \\
      1 &  -1
\end{pmatrix}.$$

The achievable rate can be easily checked with this counterexample,
and in fact there are many other examples that one can construct. We
will, in following, discuss the geometric insights that leads to these
counterexamples (and check that it is indeed a counterexample). This will also be valuable in constructing better decoders in the next section.

We first use Lemma \ref{theid} to write
$$E_{\mu_{0,\eps}} \log W_{1,\eps} \leq E_{\mu_{0,\eps}} \log W_{2,\eps} \stackrel{\text{VN}}{\longrightarrow} \|L_0-L_2\| \leq  \|L_0-L_1\| ,$$
which can be use to rewrite \eqref{eqn:case1} and \eqref{eqn:case2}
in the very noisy setting as
\begin{align}
&\text{Case 1:} \,\,\, \|L_0-L_2\|\geq \|L_0-L_1\| \label{cc1} \\
&\text{Case 2:} \,\,\, \|L_0-L_2\|\leq \|L_0-L_1\| . \label{cc2}
\end{align}
Now to construct a counterexample, we consider the special case
where $\|L_0-L_2\|= \|L_0-L_1\|$ and $\|\widetilde{L}_1\| =
\|\widetilde{L}_2\|$. These assumptions are used to simplify our
discussion, and are not necessary in constructing counterexamples.
One can check that the above example satisfies both assumptions. Now
\eqref{bug} holds if and only if
$$\frac{\langle \widetilde{L}_0,\widetilde{L}_2\rangle}{\|\widetilde{L}_2\|} \stackrel{?}{\geq} \|\widetilde{L}_2\|,$$
which is equivalent to
\begin{eqnarray*}
&\|\widetilde{L}_0\|^2 -\|\widetilde{L}_2\|^2
-\|\widetilde{L}_0-\widetilde{L}_2\|^2 \stackrel{?}{\geq} 0.
\end{eqnarray*}
It is easy to check that the last inequality does not hold for the
given counterexample, which completes the proof of Proposition
\ref{prop:badglrt}. In fact, one can write
\begin{eqnarray}
&\|\widetilde{L}_0\|^2 -\|\widetilde{L}_2\|^2 -\|\widetilde{L}_0-\widetilde{L}_2\|^2 \notag \\ &= \|\widetilde{L}_0\|^2 -\|\widetilde{L}_{1}\|^2 -\|\widetilde{L}_0-\widetilde{L}_{1}\|^2 \notag \\
&+\|\widetilde{L}_0-\widetilde{L}_{1}\|^2-\|\widetilde{L}_0-\widetilde{L}_2\|^2  \label{make} 
\end{eqnarray}
The term on the second line above is always positive (by the
one-sided property), but we have a problem with the term on the last
line: we assumed that $\|L_0-L_2\|= \|L_0-L_1\|$, and this does not
imply that
$\|\widetilde{L}_0-\widetilde{L}_{1}\|^2=\|\widetilde{L}_0-\widetilde{L}_2\|^2$
The problem here is that when using log likelihood functions as
decoding metrics, the constraints in \eqref{eqn:Ak},
\eqref{eqn:case1} and \eqref{eqn:case2} are, in the very noisy case,
given in terms of the perturbation directions $L_i, i=0,1,2$, while
the desired statement about achievable rates and the compound
capacity are given in terms of the centered directions
$\widetilde{L}_i$'s. Thus, counterexamples can be constructed by
carefully assign $\bar{L}_i$'s to be different, hence the
constraints on $L_i$'s cannot effectively regulate the behavior of
$\widetilde{L}_i$'s (\eqref{make} can be made negative). Figure \ref{counter} gives a pictorial illustration of this phenomenon.
\begin{figure}
\begin{center}
\includegraphics[scale=.8]{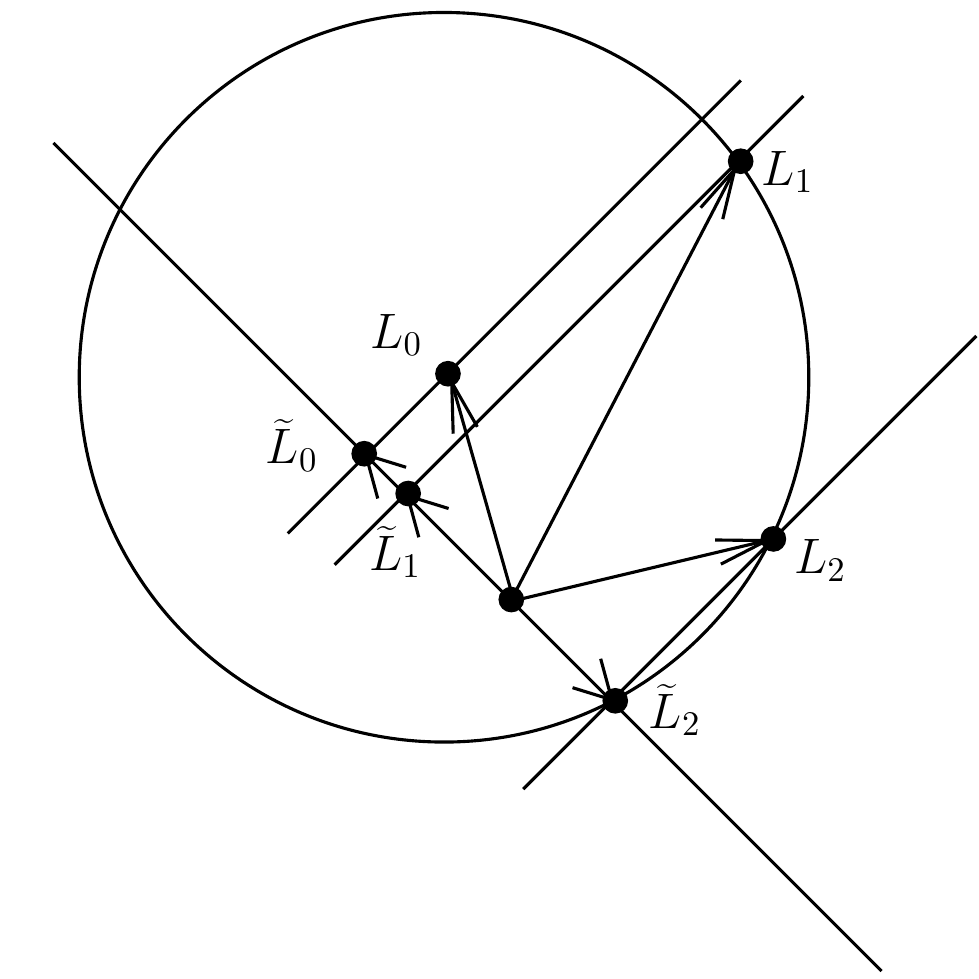}
\caption{This figure illustrates a counterexample, for binary VN channels, to the claim that GLRT with the worst channel metrics is capacity achieving. As illustrated, a  condition on the non centered directions, such as $\| L_0 - L_1\| = \|L_0 -L_2\|$, does not influence the position of the centered directions and can allow $\wtil{L}_2$ and $\wtil{L}_0$ to be opposite, violating the desired inequality in \eqref{bug}.}
\label{counter}
\end{center}
\end{figure}
The above discussion also suggests a fix to the problem. If one
could replace the constraints on $L_i$'s in
\eqref{eqn:Ak},\eqref{eqn:case1} and \eqref{eqn:case2}, by the
corresponding constraints on $\widetilde{L}_i$'s, that might at
least allow better controls over the achievable rates. This is
indeed possible by making a small change of the decoding metrics, as done in the following section.

\subsubsection{Using MAP Metrics}
We now use different metrics than the one used in previous section,
instead of the ML metrics given by $\log W_k$, we use the metrics
\begin{eqnarray}
\log \frac{W_k}{(\mu_k)_{Y}}, \label{apdef}
\end{eqnarray}
which we call the MAP metrics for maximum a posteriori
and which may also be referred as the Fano metrics in the literature. \\

As before, let us consider $W_0$, $W_1$ and $W_2$ such that $W_1$
and $W_2$ are the worst channels of two one-sided components $S_1$
and $S_2$, and $W_0$ belongs to $S_1$. Using \eqref{projs}, with
$d_1=\log \frac{W_1}{(\mu_1)_{Y}}$ and $d_2=\log
\frac{W_2}{(\mu_2)_{Y}}$, the proposed generalized linear decoder
can achieve
$$R(P_X, W_0, \{d_k\}_{k=1}^2)=R_1 \wedge R_2$$
where
$$R_k = \inf_{\A_k} D(\mu\|\mu_0^p) , \quad k=1,2$$
and for $k=1,2$
$$\A_k =\left\{\mu : \,  \mu^p=\mu_0^p, E_{\mu}  \log\frac{W_k}{(\mu_k)_{Y}} \geq   \vee_{l=1}^2E_{\mu_0} \log  \frac{W_l}{(\mu_l)_{Y}} \right\},$$
Note that again, we use same notations for this section as for the
previous one, although the sets $\A_k$ and rates $R_k$ are now given
by different expressions. Since $W_1$ and $W_2$ are the worst
channel for $P_X$ in each component, the compound capacity over
$S=S_1 \cup S_2$ is still given by
$$C(S) = I(P_X,W_1) \wedge I(P_X,W_2).$$
As we we did for \eqref{eqn:case1} and \eqref{eqn:case2}, we
consider separately two cases:
\begin{eqnarray}
\label{c1}
&\text{Case 1:} \,\,\, E_{\mu_0}  \frac{W_1}{(\mu_1)_{Y}} \geq E_{\mu_0}  \frac{W_2}{(\mu_2)_{Y}}  \\
\label{eqn:mapvnc2}
&\text{Case 2:} \,\,\, E_{\mu_0}  \frac{W_1}{(\mu_1)_{Y}} \leq
E_{\mu_0}  \frac{W_2}{(\mu_2)_{Y}}.
\end{eqnarray}

Following the same argument as in the last section, we verify that
$R_1\geq C(S)$ under both cases. Note that in case 1, the constraint in $\A_1$ is
$E_\mu\log\frac{W_1}{(\mu_1)_{Y}} \geq E_{\mu_0} \log
\frac{W_1}{(\mu_1)_{Y}}$. Comparing this with its counterpart for
in ML decoding, the only difference is the extra $E\log(\mu_1)_Y$
terms on both sides. Noticing that $\mu$ and $\mu_0$ have the same
$Y$ marginal distribution, we see that the optimization problem is
exactly the same as before, and thus the achievable rate is the
mismatched rate $R(P_X, W_0, \log W_1)$, which by the one-sided
assumption $W_0\in S_1$ is higher than $I(P_X, W_1)$. In case 2,
$R_1> C(S)$ follows since \eqref{eqn:mapvnc2} gives a more stringent
constraint in $\A_1$, and hence a higher achievable rate (conf.
\eqref{eqn:R1case2}). Hence, just like it was the case for the ML decoding metrics, 
$R_1> C(S)$ is easily checked with the one-sided property. We now show that as opposed to the ML case, with the MAP metrics, we also have $R_2 \geq C(S)$.

The main difference between the proposed MAP decoding metric and the
ML metric used in the previous section can be seen clearly from the
very noisy setting. Using a similar argument as in Lemma
\ref{theid}, we have
\begin{align}
 &E_{\mu_{0,\eps}} \log \frac{W_{1,\eps}}{ (\mu_{1,\eps})_Y} \notag \\
 &\vnl \langle \widetilde{L}_0,\widetilde{L}_1  \rangle - \frac{1}{2} \|\widetilde{L}_1\|^2= \frac{1}{2}(\|\widetilde{L}_0\|^2-\|\widetilde{L}_0-\widetilde{L}_1\|^2).
 \label{useit}
\end{align}
Thus, the optimization in $R_k$ are over the sets
\begin{eqnarray}
\label{eqn:vngamp}
&&A_{k, \epsilon} = \{L: \bar{L}= \bar{L}_0:\\
&&\quad\langle\widetilde{L},
\widetilde{L}_k\rangle -\frac{1}{2} \|\widetilde{L}_k\|^2\geq
\vee_{l=1}^2
 \frac{1}{2} (\|\widetilde{L}_0\|^2-
\|\widetilde{L}_0-\widetilde{L}_l\|^2) \}\nonumber
\end{eqnarray}
and the two cases to be considered are

\begin{eqnarray}
\label{eqn:gmapvncase1}
&\text{Case 1:} \,\,\, \|\widetilde{L}_0-\widetilde{L}_1\|^2 \leq \|\widetilde{L}_0-\widetilde{L}_2\|^2 \\
\label{eqn:gmapvncase2}
&\text{Case 2:} \,\,\,  \|\widetilde{L}_0-\widetilde{L}_1\|^2 \geq
\|\widetilde{L}_0-\widetilde{L}_2\|^2  .
\end{eqnarray}
These expressions are almost the same as the ones for the ML metric,
the very noisy version of \eqref{eqn:Ak}, \eqref{cc1}, and
\eqref{cc2}, except now we have the conditions on the centered
directions (tilde vectors). As discussed in the proof of Proposition
\ref{prop:badglrt}, this change is precisely what is needed to avoid
the counter example. It turns out that this change is also
sufficient for the decoder to be capacity achieving.

Now what remains to be proved is that $R_2 \geq C(S)$. Using
\eqref{useit}, and noticing the marginal constraints, we have for case 1
\begin{eqnarray*}
R_{2,\eps}  \vnl
 \min_{L: \,  \bar{L}=\bar{L}_0, \langle \widetilde{L},\widetilde{L}_2 \rangle  \geq \frac{1}{2}(\|\widetilde{L}_0\|^2+\|\widetilde{L}_2\|^2 - \|\widetilde{L}_0-\widetilde{L}_1\|^2 )} \|\widetilde{L}\|^2
\end{eqnarray*}
and for case 2
\begin{eqnarray*}
R_{2,\eps}  \vnl
\min_{L: \,  \bar{L}=\bar{L}_0, \langle \widetilde{L},\widetilde{L}_2 \rangle  \geq \frac{1}{2}(\|\widetilde{L}_0\|^2+\|\widetilde{L}_2\|^2 - \|\widetilde{L}_0-\widetilde{L}_2\|^2 )} \|\widetilde{L}\|^2 .
\end{eqnarray*}

These optimizations can be explicitly solved as projections:
\begin{align*}
 & \text{For Case 1:}\,\,\, R_{2,\eps} \vnl  \frac{1}{4} \frac{(\|\widetilde{L}_0\|^2+\|\widetilde{L}_2\|^2 - \|\widetilde{L}_0-\widetilde{L}_1\|^2 )^2 }{\| \wtil{L}_2 \|^2} \\
 & \text{For Case 2:}\,\,\, R_{2,\eps}  \vnl  \frac{1}{4} \frac{(\|\widetilde{L}_0\|^2+\|\widetilde{L}_2\|^2 - \|\widetilde{L}_0-\widetilde{L}_2\|^2 )^2 }{\| \wtil{L}_2 \|^2}.
\end{align*}
Recalling that the compound capacity is given by
\begin{align*}
&C(S_\eps) \vnl  \|\widetilde{L}_1\|^2 \wedge \|\widetilde{L}_2\|^2,
\end{align*}
we have
\begin{align}
 & \text{Case 1:}\,\,\, R_{2,\eps} \geq C(S_\eps) \notag \\
 &\vnl  \frac{1}{2} \frac{\|\widetilde{L}_0\|^2+\|\widetilde{L}_2\|^2 - \|\widetilde{L}_0-\widetilde{L}_1\|^2  }{\| \wtil{L}_2 \|} \geq \|\widetilde{L}_1\| \wedge \|\widetilde{L}_2\| \label{f1} \\
 & \text{Case 2:}\,\,\, R_{2,\eps} \geq C(S_\eps) \notag \\
 &\vnl  \frac{1}{2} \frac{\|\widetilde{L}_0\|^2+\|\widetilde{L}_2\|^2 - \|\widetilde{L}_0-\widetilde{L}_2\|^2  }{\| \wtil{L}_2 \|} \geq \|\widetilde{L}_1\| \wedge \|\widetilde{L}_2\| \label{f2}
\end{align}
and we now check that inequalities \eqref{f1} and \eqref{f2} hold with
$\|\widetilde{L}_1\|$ instead of $\|\widetilde{L}_1\| \wedge
\|\widetilde{L}_2\|$ on the right hand side.\\
Starting with \eqref{f2}, we write
\begin{align*}
&\frac{\frac{1}{2} (\|\widetilde{L}_0\|^2+\|\widetilde{L}_2\|^2 - \|\widetilde{L}_0-\widetilde{L}_2\|^2)  }{\|\widetilde{L}_2\|} - \|\widetilde{L}_1\| =\\
& \frac{\frac{1}{2} (\|\widetilde{L}_0\|^2+\|\widetilde{L}_2\|^2  -2 \|\widetilde{L}_1\|\|\widetilde{L}_2\| - \|\widetilde{L}_0-\widetilde{L}_2\|^2)  }{\|\widetilde{L}_2\|}=\\
& \frac{\frac{1}{2} ((\|\widetilde{L}_1\|-\|\widetilde{L}_2\|)^2 + \|\widetilde{L}_0\|^2-\|\widetilde{L}_1\|^2   - \|\widetilde{L}_0-\widetilde{L}_2\|^2)  }{\|\widetilde{L}_2\|} \stackrel{\eqref{eqn:gmapvncase2}}{\geq}\\
& \frac{\frac{1}{2} ((\|\widetilde{L}_1\|-\|\widetilde{L}_2\|)^2 + \|\widetilde{L}_0\|^2-\|\widetilde{L}_1\|^2   - \|\widetilde{L}_0-\widetilde{L}_1\|^2)
}{\|\widetilde{L}_2\|} \geq0
\end{align*}
where last inequality follows from the one-sided property
\begin{eqnarray*}
\frac{\langle \widetilde{L}_0,\widetilde{L}_1\rangle^2}{\|\widetilde{L}_1\|^2} \geq \|\widetilde{L}_ 1\|^2  \Leftrightarrow
\|\widetilde{L}_0\|^2 -\|\widetilde{L}_1\|^2
-\|\widetilde{L}_0-\widetilde{L}_1\|^2  \geq 0
\end{eqnarray*}
For \eqref{f1}, the same expansion gets us directly to
\begin{align*}
  & \frac{\frac{1}{2} (\|\widetilde{L}_0\|^2+\|\widetilde{L}_2\|^2 - \|\widetilde{L}_0-\widetilde{L}_1\|^2 ) }{\| \wtil{L}_2 \|} -\|\widetilde{L}_1\|  = \\
 & \frac{\frac{1}{2} ((\|\widetilde{L}_1\|-\|\widetilde{L}_2\|)^2
+ \|\widetilde{L}_0\|^2-\|\widetilde{L}_1\|^2   -
\|\widetilde{L}_0-\widetilde{L}_1\|^2)  }{\|\widetilde{L}_2\|} \geq 0
\end{align*}
again by the one-sided property. Now combining these results, we get
that the GMAP decoder is capacity achieving for the VN
case. The result can be easily generalized to cases with more than two one-sided components.

\vspace{.5cm}
\noindent{\it Discussions:}

The above derivations can also be viewed from a pictorial way. We
take case 2 for $R_2$ for example. The one-sided constraint
$\langle\widetilde{L}_0, \widetilde{L}_1\rangle \geq
\|\widetilde{L}_1\|^2$ says that $\widetilde{L}_0$ lies on the right
side of $\widetilde{L}_1$; but the constraint for
case 2 , \eqref{eqn:gmapvncase2}, precisely implies that $\widetilde{L}_2$ can
only lie in the smaller circle centered at $\widetilde{L}_0$, as in
Figure \ref{diamondA}, but the small circle intersect the large circle only in the hatched region, where 
\begin{align}
\frac{\langle \widetilde{L}_0,\widetilde{L}_2\rangle}{\|\widetilde{L}_2\|} \geq \|\widetilde{L}_1\| \wedge
\|\widetilde{L}_2\|,\label{goodbad}
\end{align}
holds. On the other hand, if we work with the ML metrics, the constraint for case 2 is given by \eqref{cc2}, and how we showed it in the counterexample of section \ref{ml}, this does no longer force $\widetilde{L}_2$ to lie inside the smaller circle centered at $\widetilde{L}_0$, hence inside the hatched region, as Figure \ref{diamondB} and \ref{3D} illustrates it.
\begin{figure}
\begin{center}
\includegraphics[scale=.56]{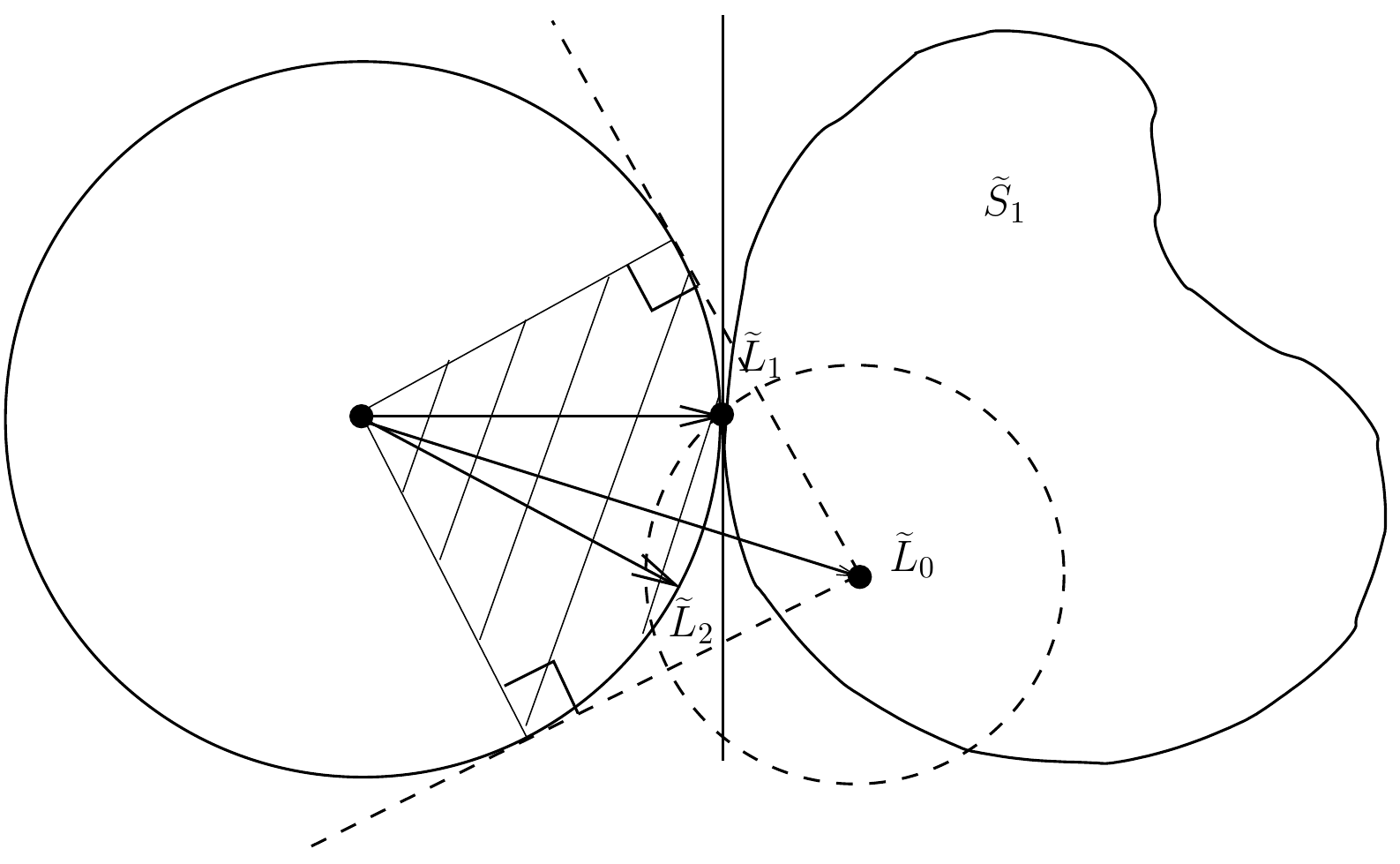}
\caption{Location of $\wtil{L}_2$ where \eqref{goodbad} holds.}
\label{diamondA}
\end{center}
\end{figure}
\begin{figure}
\begin{center}
\includegraphics[scale=.56]{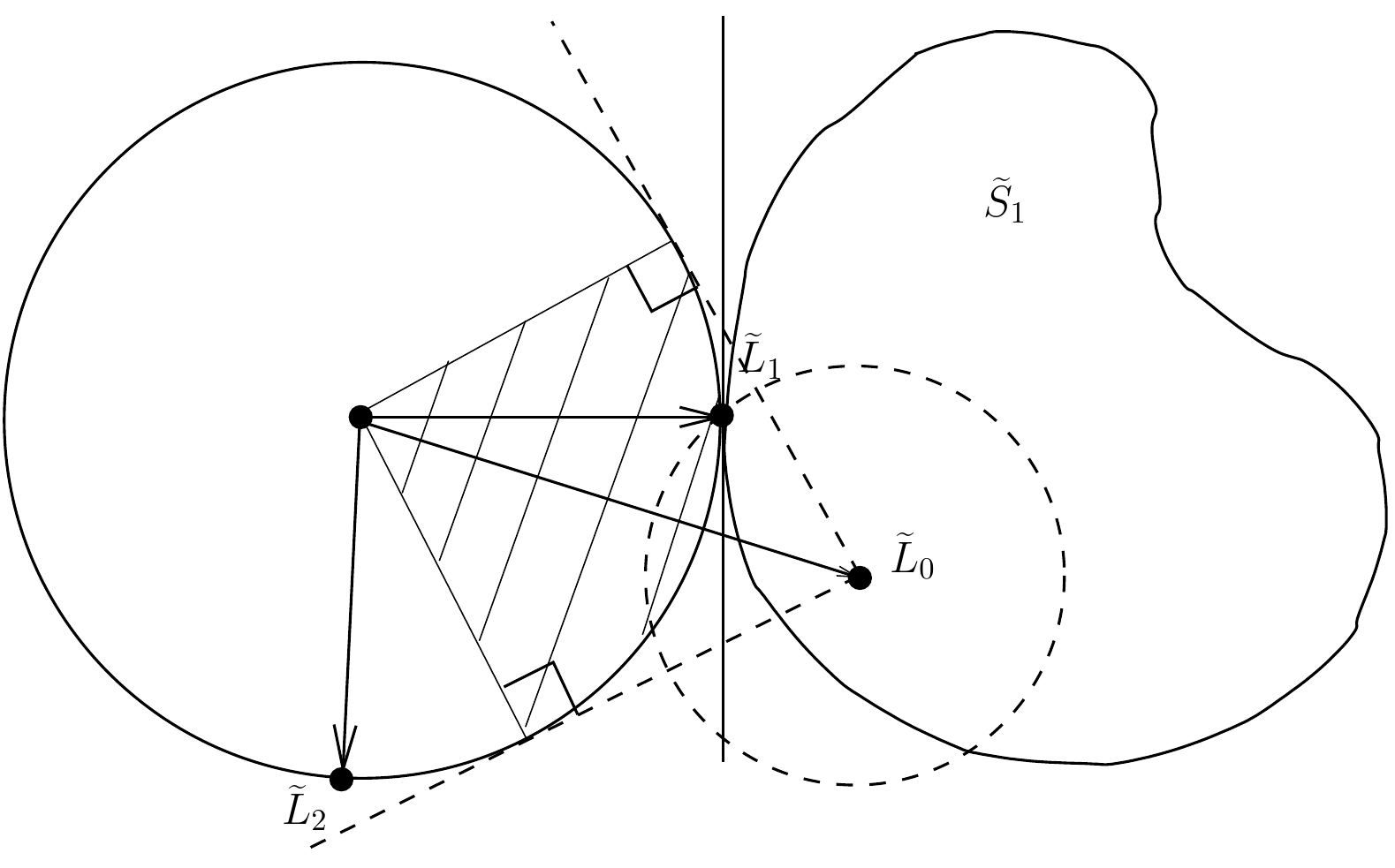}
\caption{Location of $\wtil{L}_2$ where \eqref{goodbad} does not hold.}
\label{diamondB}
\end{center}
\end{figure}
\begin{figure}
\begin{center}
\includegraphics[scale=.54]{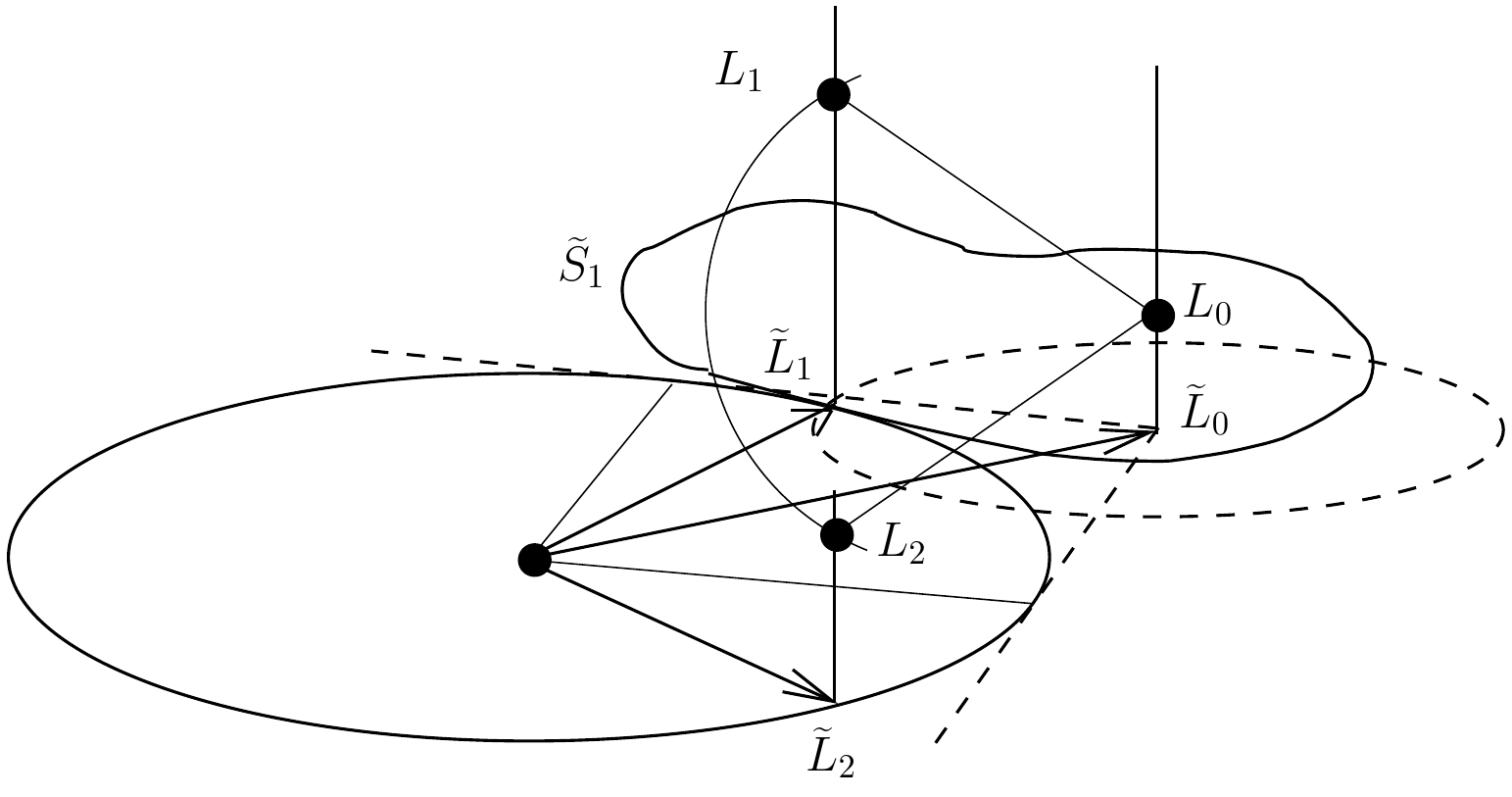}
\caption{This figure illustrates that on a 3-ary input/output VN channels, the non centered directions (living in the 3D space) verify $\| L_0 -L_1\| = \| L_0 -L_2\|$, but this does not influence the position of the centered directions (in the 2D plane) and indeed $\wtil{L}_2$ is in the region where \eqref{goodbad} does not hold, as in Figure \ref{diamondB}.}
\label{3D}
\end{center}
\end{figure}

It is insightful to try to understand the reason that the GMAP
decoder works well while the GLRT fails. For a linear decoder with a
single metric $d:\X\times \Y \mapsto \mR$, if one forms a different
test by picking $d'(x,y) = d(x,y) + f(y)$, for some function $f:\Y
\to \mR$, it is not hard to see that the resulting decision is exactly the
same, for every possible received signal $y$. This is why the ML
decoder and the MAP decoder, from the same mismatched channel $W_1$,
are indeed equivalent, as they differ by a factor of $f=\log (P_X
\circ W_1)_Y$. For a generalized linear decoder with multiple
metrics, $d_1, d_2, \ldots, d_K$, if one changes the metrics to
$d_1+f, d_2+f, \ldots, d_K+f$, for the same function $f$ on $\Y$,
again the resulting decoder is the same. Things are different,
however, if one changes these metrics by different functions, to
have $d_1+f_1, \ldots, d_K+f_K$. The problem is that this changes
the balance between the metrics, which as we observed in the GMAP
story, is critical for the generalized linear decoder to work
properly. For example, if one adds a big number on one of the
metrics to make it always dominate the others, the purpose of using
multiple metrics is defeated. GLRT differs from the GMAP decoder by
factors of $\log (\mu_k)_Y$ on the $k^{th}$ metric, which causes a
bias depending the received signal $y$. The counter example we
presented in the precious section is in essence constructed to
illustrate the effect of such bias. Through a similar approach, one
can indeed show that the GMAP receiver is the {\it unique}
generalized linear receiver, based on the worst channels of
different one-sided components, in the sense that any non-trivial
variation of these metrics, i.e., $f_1, \ldots, f_K$ which are not
the same function, would result in a receiver that does not achieve
the compound capacity in all cases. Counter examples can always be
constructed in a similar fashion.

\section{Linear Decoding for Compound Channel:\\ The General Case}
\label{sec:lift}
\subsection{The Results}
The previous section gives us a series of results regarding linear
decoders on different kinds of compound sets, in the very noisy
setting. While focusing on special channels, the geometric insights we
developed in the previous section is clearly helpful in
understanding the problem in general. In this section, we will show
that indeed most of the results reported in the previous section 
have ``natural" counterparts in the general not very noisy cases.
Moreover, the proofs of these general results often proceed in a
step by step correspondence with that for the very noisy case. We
often refer to such procedure of generalizing the results from the
very noisy case to the general cases, as ``lifting". In the
following, we will first list all the general results, and give
proofs in section \ref{lifting}.

Recall the optimal input distribution of a set $S$ by
$$P_X=\arg \max_{P \in M_1(\X)} \inf_{W \in S} I(P,W),$$
and if the maximizers are not unique, we define $P_X$ to be any arbitrary maximizer.
\begin{definition}\label{onesideddef} {\bf One-sided Set}\\
A set $S$ is one-sided, if
\begin{eqnarray}
D(\mu_0\|\mu_S^p)\geq D(\mu_0\|\mu_S) + D(\mu_S\|\mu_S^p) , \quad \forall W_0 \in S. \label{onesided}
\end{eqnarray}
where
\begin{eqnarray}
W_S=\arg \min_{W \in \mathrm{cl}(S)} I(P_X,W). \label{worsechannel}
\end{eqnarray}
and $\mu_0= P_X \circ W_0, \mu_S = P_X \circ W_S$, are the joint distribution over the channel $W_0$ and $W_S$, respectively.
\end{definition}
\noindent
Note that in order for \eqref{onesided} to hold, the minimizer in \eqref{worsechannel} must be unique.

\begin{proposition}\label{onesidedproposition}
For one-sided sets $S$, the linear decoder induced by the metric $d=\log W_S$ is capacity achieving.
\end{proposition}
Note that in \cite{cisnar}, the same linear decoder is proved to be
capacity achieving for the case where $S$ is convex.

\begin{proposition}\label{charact}
Convex sets are one-sided and there exist one-sided sets that are not convex.
\end{proposition}

\begin{proposition}\label{finitesetproposition}
For any set $S$, the decoder maximizing the score function
$G_n=\sup_{W \in S} \log W^n$,  is capacity achieving, but
generalized linear only if $S$ is finite.
\end{proposition}

\begin{proposition}\label{finiteunionml}
For $S = \cup_{k=1}^K S_k$, where $\{S_k\}_{k=1}^K$ are one-sided
sets, the generalized linear decoder induced by the metrics
$d_k=\log W_{S_k}$, for $1 \leq k \leq K$, is not capacity achieving
(in general).
\end{proposition}

The following Theorem is the main result of the paper.
\begin{theorem}\label{finiteunionap}
For $S = \cup_{k=1}^K S_k$, where $\{S_k\}_{k=1}^K$ are one-sided
sets, the generalized linear decoder induced by the metrics
$d_k=\log \frac{W_{S_k}}{(\mu_{S_k})_Y}$, for $1 \leq k \leq K$ is
capacity-achieving.
\end{theorem}

\vspace{1cm}


\subsection{Proofs: Lifting Local to Global Results}\label{lifting}
In this section, we illustrate how the results and proofs obtained
in section \ref{ludvn} in the very noisy setting can be lifted to
results and proofs in the general setting. We first consider the
case of one-sided sets. By revisiting the definitions made in
section \ref{onesidedsec}, we will try to develop a ``naturally"
corresponding notion of one-sidedness for the general problems.

By definition of a VN one-sided set, $\Sc$ is such that
\begin{eqnarray}
\|\widetilde{L}_0\|^2 -\|\widetilde{L}_ {\Sc}\|^2
-\|\widetilde{L}_0-\widetilde{L}_ {\Sc}\|^2  \geq 0, \quad  \forall
L_0 \in \Sc. \label{onesidednorm}
\end{eqnarray}
Next, we find the divergences, for the general problems, whose very
noisy representations are these norms: recall that
\begin{eqnarray}
\label{eqn:correspond1}
 D(\mu_0\|\mu_0^p) \vnl  \| L_0 -\bar{L}_0 \|^2=\| \wtil{L}_0
\|^2
\end{eqnarray}
and
\begin{eqnarray}
\label{eqn:correspond2}
 D(\mu_S\|\mu_S^p) \vnl  \| L_S -\bar{L}_S \|^2=\|
\wtil{L}_S \|^2.
\end{eqnarray}

On the other hand, we also have
$$ D(\mu_0\|\mu_S) \vnl   \| L_0 -L_S \|^2$$
and
$$ D(\mu_0^p\|\mu_S^p) \vnl   \| \bar{L}_0 -\bar{L}_S \|^2,$$
and hence
\begin{eqnarray}
&& D(\mu_0\|\mu_S) -  D(\mu_0^p\|\mu_S^p) \vnl \label{logsumlift} \\
&&\qquad \| L_0 -L_S \|^2-\| \bar{L}_0 -\bar{L}_S \|^2= \|
\wtil{L}_0 -\wtil{L}_S \|^2 \notag
\end{eqnarray}
where the last equality simply uses the projection principle, i.e.,
that the projection of $L$ onto the centered directions given by $\wtil{L}=L-\bar{L}$, is orthogonal to the projection's height
$\bar{L}$, implying $$ \lVert \wtil{L} \rVert^2= \lVert L \rVert^2-
\lVert \bar{L} \rVert^2 .$$

Now, by reversing the very noisy approximation in
\eqref{eqn:correspond1}, \eqref{eqn:correspond2} and
\eqref{logsumlift}, we get that
\begin{eqnarray*}
 D(\mu_0\|\mu_0^p) - D(\mu_S\|\mu_S^p) - (D(\mu_0\|\mu_S) -  D(\mu_0^p\|\mu_S^p) ) \geq 0
\end{eqnarray*}
for all $W_0\in S$, can be viewed as a ``natural" counterpart
of \eqref{defonesided}, hence of the VN one-sided definition. With a little simplification, this
inequality is equivalent to
\begin{eqnarray}
 D(\mu_0\|\mu_S^p) \geq D(\mu_0\|\mu_S)+ D(\mu_S\|\mu_S^p), \qquad \forall W_0 \in S. \label{liftonesided}
\end{eqnarray}
Therefore, we use this as the definition of the general one-sided sets, as expressed  in Definition \ref{onesideddef}.

Clearly, as we mechanically generalized the notion of one-sided sets
from a special very noisy case to the general problem, there is no
reason to believe at this point that the resulting one-sided sets will
have the same property in the general setting, than their counterparts in the very noisy case; namely, that the linear decoder induced from the worst channel achieves the compound capacity. However, this turns out to be true, and the proof again
follows closely the corresponding proof of the very noisy
special case.

\vspace{0.5cm}

\begin{proof} {\it of Proposition \ref{onesidedproposition}.}

Recall that in the VN case, when the actual channel is $W_{0,
\epsilon}$, and the decoder uses metric $d_\epsilon= \log W_{1, \epsilon}$, the
achievable rate, in terms of the corresponding centered directions
$\widetilde{L}_0,
\widetilde{L}_1$, is given by, cf. \eqref{optim},
\begin{eqnarray}
\lim_{\epsilon \to 0} \frac{2}{\epsilon^2} R( P_X, W_{0, \epsilon},
d_{\epsilon}) = \inf_{\widetilde{L}: \langle \widetilde{L},
\widetilde{L_1}\rangle \geq \langle\widetilde{L}_0,
\widetilde{L}_1\rangle} \|\widetilde{L}\|^2
\end{eqnarray}

The constraint of the optimization can be rewritten in norms as
\begin{eqnarray}
\label{eqn:vnnorm}
\widetilde{L}: \|\widetilde{L}\|^2 -
\|\widetilde{L}-\widetilde{L}_1\|^2 \geq \|\widetilde{L}_0\|^2 -
\|\widetilde{L}_0-\widetilde{L}_1\|^2
\end{eqnarray}
Now if $\widetilde{L}_0$ lies in a one-sided set $S$, and we use
decoding metric as the worst channel $\widetilde{L}_1 =
\widetilde{L}_S$, by using definition \eqref{onesidednorm}, and
recognizing that $\|\widetilde{L}-\widetilde{L}_1\|^2$ is
non-negative, this constraint implies
\begin{eqnarray}
\|\widetilde{L}\|^2 \geq \|\widetilde{L}_0\|^2 -
\|\widetilde{L}_0-\widetilde{L}_S\|^2 \geq \|\widetilde{L}_S\|^2,
\quad \forall \widetilde{L}_0\in S, \label{repl}
\end{eqnarray}
form which we conclude that the compound capacity is achievable. The
proof of Proposition \ref{onesidedproposition} replicates these steps
closely.

First, we write in the general setting, the mismatched mutual
information is given by
 \begin{eqnarray}
R(P_X, W_0, \log W_S)=  \inf_{\mu \in \A_S} D(\mu\|\mu^p) \label{last}
\end{eqnarray}
where $$\A_S=\{  \mu_X=P_X, \mu_Y=(\mu_0)_Y ,  E_{\mu} \log W_S \geq
E_{\mu_0} \log W_S\}.$$

Since we consider here a linear decoder, i.e. induced by only one single-letter metric, we can consider equivalently the ML or MAP metrics. We then work with the MAP metric and the constraint set is equivalently expressed as:
$$\A_S=\{  \mu_X=P_X, \mu_Y=(\mu_0)_Y ,  E_{\mu} \log \frac{W_S}{ (\mu_S)_Y} \geq E_{\mu_0} \frac{W_S}{ (\mu_S)_Y}\}.$$
Expressing the quantities of interest in terms of divergences, we
write
\begin{eqnarray*}
&&E_{\mu} \left[ \frac{W_S}{ (\mu_S)_Y} \right]\\
&=& E_{\mu} \left[ \log \frac{W_S}{(\mu_S)_Y}
\frac{\mu}{\mu} \frac{\mu^p}{\mu^p}\right] \\
&=& D(\mu\|\mu^p) - D(\mu\| \mu_S) + D(\mu^p\|\mu_S^p)
\end{eqnarray*}
Similarly we have
$$E_{\mu_0} \frac{W_S}{ (\mu_S)_Y} =D(\mu_0\|\mu_0^p) - D(\mu_0\| \mu_S) + D(\mu_0^p\|\mu_S^p).$$
%
Thus we can rewrite $\A_S$ as
\begin{eqnarray}
\label{eqn:gennorm}
\A_S &=& \{ \mu: \mu_X= P_X, \mu_Y = (\mu_0)_Y \nonumber\\
&&D(\mu\|\mu^p)-D(\mu\|\mu_S)+D(\mu^p\|\mu_S^p)  \nonumber\\
&& \geq D(\mu_0\|\mu_0^p)-D(\mu_0\|\mu_S)+D(\mu_0^p\|\mu_S^p) \}
\end{eqnarray}
It worth noticing that this is precisely the lifting of
\eqref{eqn:vnnorm}.

Now, in the VN limit, $D(\mu\|\mu_S)-D(\mu^p\|\mu_S^p)$ is given by
$\| L - L_S \|^2-\| \bar{L} - \bar{L}_S \|^2=\| \wtil{L} - \wtil{L}_S \|^2$, which is clearly positive. Here, we have that
$$D(\mu\|\mu_S)-D(\mu^p\|\mu_S^p)\geq 0,$$
is a direct consequence of log-sum inequality, and with this,
we can write for all $\mu \in \A_S$,
\begin{eqnarray*}
D(\mu\|\mu^p) \geq
D(\mu_0\|\mu_0^p)-D(\mu_0\|\mu_S)+D(\mu_0^p\|\mu_S^p)
\end{eqnarray*}
which is in turn lower bounded by $D(\mu_S\|\mu_S^p)= I(P_X, W_S)$,
provided that the set $S$ is one-sided, cf. \eqref{onesideddef} (note that last lines are again a lifting of \eqref{repl}).
Thus, the compound capacity is achieved.
\end{proof}

This general proof can indeed be shortened. Here, we emphasize the
correspondence with the proof for the very noisy case, in order to
demonstrate the insights one obtains by using the local
geometric analysis.

\vspace{0.5cm}

\begin{proof}{\it of Lemma \ref{charact}.}

Let $C$ a convex set, then for any input distribution $P_X $ the set
$D=\{\mu |
\mu(a,b)=P_X(a)W(b|a), W
\in C \}$ is a convex set as well. For $\mu$ such that
$\mu(a,b)=P_X(a)W(b|a)$, we have
$$D(\mu \| \mu_C^p)=I(P_X,W)+D(\mu_Y\|(\mu_C)_Y),$$
hence we obtain, by definition of $W_C$ being the worse channel of $\mathrm{cl}(C)$,
$$\mu_C=\min_{\mu \in \mathrm{cl}(D)} D(\mu \| \mu_C^p).$$
Therefore, we can use theorem 3.1. in \cite{cistut} and for any $\mu_0 \in D$, we have
the pythagorean inequality for convex sets
\begin{eqnarray}
D(\mu_0\|\mu_C^p)\geq D(\mu_0\|\mu_C) + D(\mu_C\|\mu_C^p). \label{pyth}
\end{eqnarray}
This concludes the proof of the first claim of the Proposition. Now
to construct a one-sided set that is not convex, one can simply take
a convex set and remove one point in the interior, to create a
"hole". This does not affect the one-sidedness, but makes the set
non-convex. It also shows that there are sets that are one-sided (and not convex) for all input distributions, so the one-sidedness does not have to depend on which input distribution is chosen.
\end{proof}
Proposition
\ref{charact} says that our definition of one-sided sets is strictly
more general than convex sets. This generalizes the known
result \cite{cisnar} on when does linear receiver achieve compound
capacity, but more importantly, our definition leads to the
meaningful use of generalized linear decoders with finite number of
metrics: it is easy to construct an example of compound set with an
infinite number of disconnected convex components; but the notion of
finite unions of one-sided sets is general enough to include most
compound sets that one can be exposed to. 

In the next proofs, we no longer give explicitly the analogy with the VN setting.
\vspace{0.5cm}

%

\begin{proof}{\it of Proposition \ref{finitesetproposition}.}

We need to show the following
\begin{align*}
&\wedge_{W_1 \in S} \inf_{\mu : \,\mu^p = \mu_0^p, E_{\mu} \log W_1 \geq  \vee_{W \in S} E_{\mu_0} \log W} D(\mu\|\mu_0^p) \\&  \geq \wedge_{W \in S} I(P_X,W),
\end{align*}
and we will see that the left hand side of this inequality is equal
to $I(P_X,W_0)$. Note that $\vee_{W \in S}  E_{\mu_0} \log
W=E_{\mu_0} \log W_0 = I(P_X,W_0)$. Thus, the desired inequality is
equivalent to $\forall W_1 \in S$,
\begin{eqnarray}
\inf_{\mu : \, \mu^p=\mu_0^p, E_{\mu} \log W_1 \geq E_{\mu_0} \log W_0} D(\mu\|\mu_0^p) \geq \wedge_{W \in S} I(P_X,W).
\label{ineq2}
\end{eqnarray}
Using the marginal constraint $\mu^p=\mu_0^p$, we
have
\begin{eqnarray}
&&E_{\mu} \log W_1 \geq E_{\mu_0} \log W_0  \notag \\
&\Leftrightarrow&E_\mu\left[\log \frac{\mu_1}{\mu^p}\right] \geq
E_{\mu_0}\left[\log \frac{\mu_0}{\mu_0^p}\right] \notag\\
&\Leftrightarrow& D(\mu\|\mu^p) - D(\mu\| \mu_1) \geq D(\mu_0\|
\mu_0^p)
\label{tighter}
\end{eqnarray}
using the fact that $D(\mu\|\mu_1)\geq 0$, we have
\begin{align}
&\inf_{\mu : \,  \mu^p=\mu_0^p,  E_{\mu} \log W_1 \geq E_{\mu_0} \log W_0} D(\mu\|\mu_0^p) \notag \\
&= \inf_{\mu : \,  \mu^p =\mu_0^p, D(\mu\|\mu^p) - D(\mu\|\mu_1)\geq D(\mu_0\|\mu_0^p)} D(\mu\|\mu^p) \label{altern} \\
&\geq D(\mu_0\|\mu_0^p) = I(P_X,W_0).
\end{align}
This concludes the proof of the Proposition. In fact, one could get
a tighter lower bound by expressing
\eqref{tighter} as
\begin{align*}
&E_{\mu} \log W_1 \geq E_{\mu_0} \log W_0  \, \Leftrightarrow \\
& D(\mu\|\mu^p) - (D(\mu\|\mu_1)- D(\mu^p\|\mu_1^p))\\
&\geq D(\mu_0\|\mu_0^p) +D(\mu_0^p\|\mu_1^p),
\end{align*}
and using the log-sum inequality to show that $D(\mu\|\mu_1)- D(\mu^p\|\mu_1^p) \geq 0$, \eqref{altern} is lower bounded
by
$$D(\mu_0\|\mu_0^p) +D(\mu_0^p\|\mu_1^p).$$
Figure \ref{2planes} illustrates this gap.
\end{proof}
\begin{figure}
\begin{center}
\includegraphics[scale=1]{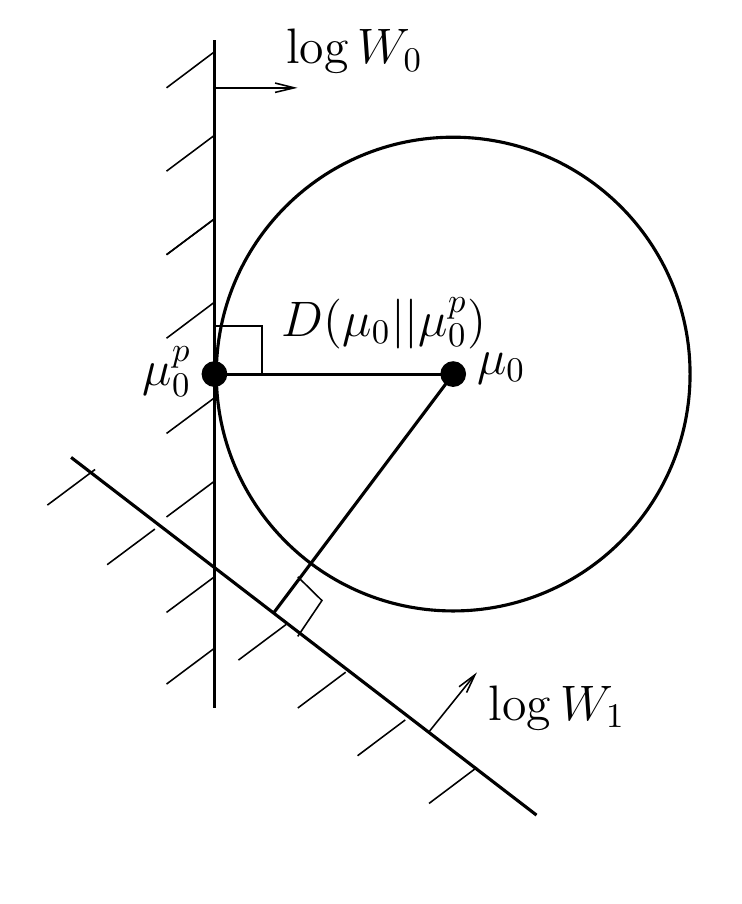}
\caption{This figure represent the left hand side of \eqref{ineq2}. It indeed represents two cases: when $W_1=W_0$ and when $W_1$ is an arbitrary channel in $S$.
The planes in the figure represent the constraint sets appearing in the optimization  for each of these cases. The fact that the twisted plane is not tangent to the divergence ball with radius $D(\mu_0 \| \mu_0^p)$ illustrates the gap pointed out in the proof of Proposition \eqref{finitesetproposition}.
}
\label{2planes}
\end{center}
\end{figure}
\begin{proof}{\it of Proposition \ref{finiteunionml}.} 

We found a counter-example for
the very noisy setting in section \ref{uniononesidedsec}, therefore
the negative statement holds in the general setting.
\end{proof}

\vspace{0.5cm}

\begin{proof}{\it of Theorem\ref{finiteunionap}.}

We need to show
\begin{eqnarray}
 \inf_{\mu \in \A} D(\mu\|\mu_0^p)
  \geq \wedge_{k=1}^K  I(P_X,W_k), \label{main}
\end{eqnarray}
where $\A$ contains all joint distributions $\mu$ such that
\begin{eqnarray}
\mu_X=P_X,\,\,\, \mu_Y=(\mu_0)_Y, \label{c1}\\
\vee_{k=1}^K E_{\mu} \log \frac{W_k}{ (\mu_k)_Y} \geq \vee_{k=1}^K E_{\mu_0} \log \frac{W_k}{ (\mu_k)_Y} . \label{c2}
\end{eqnarray}
We can assume w.l.o.g. that $W_0 \in C_1$.
We then have
\begin{eqnarray*}
D(\mu\|\mu_0^p) &\stackrel{\text{(A)}}{=}& D(\mu\|\mu^p)  \\
& \stackrel{\text{(B)}}{\geq} & \vee_{k=1}^K E_{\mu} \log \frac{W_k}{ (\mu_k)_Y} \\
&\stackrel{\text{(C)}}{\geq} & \vee_{k=1}^K E_{\mu_0} \log \frac{W_k}{ (\mu_k)_Y}  \\
&\geq& E_{\mu_0} \log \frac{W_1}{(\mu_1)_Y} \\
&\stackrel{\text{(D)}}{\geq} & E_{\mu_1} \log \frac{W_1}{(\mu_1)_Y} \\
&=& I(P_X,W_1)\\
&\geq& \wedge_{k=1}^K  I(P_X,W_k),
\end{eqnarray*}
where (A) uses \eqref{c1},  (B) uses the log-sum inequality:
\begin{eqnarray*}
 E_{\mu} \log \frac{W_k}{ (\mu_k)_Y} &=& D(\mu\|\mu^p)  + E_{\mu} \log \frac{W_k}{ (\mu_k)_Y} - D(\mu\|\mu^p) \\
 &= &  D(\mu\|\mu^p)-  \underbrace{( D(\mu\|\mu_k)- D(\mu^p\|\mu_k^p))}_{\geq 0 },
\end{eqnarray*}
(C) is simply \eqref{c2} and (D) follows from the one-sided
property:
\begin{eqnarray*}
 && E_{\mu_0} \log \frac{W_1}{(\mu_1)_Y} - E_{\mu_1} \log \frac{W_1}{(\mu_1)_Y}\\
&=& D(\mu_0\|\mu_1^p)-   D(\mu_0\|\mu_1) - D(\mu_1\|\mu_1^p) \\
&\geq& 0.
\end{eqnarray*}
\end{proof}

\subsection{Discussions}

We raised the question whether it is possible for a decoder to be both linear and
capacity achieving on compound channels.
We showed that if the compound set is a union of one-sided sets, a
generalized linear which is capacity achieving decoder exists. We constructed it as
follows:  if $W_1,\ldots,W_K$ are the worst channels of each
component (cf. figure \ref{algo}), use the generalized linear
decoder induced by the MAP metrics $\log
\frac{W_1}{(\mu_1)_Y},\ldots,\log \frac{W_K}{(\mu_K)_Y}$, i.e.,
decode with

$$ G_n(y)=\arg\max_{m\in\{1,\ldots,M\}} \vee_{k=1}^K E_{\hat{P}_{(x_m,y)}} \log \frac{W_k}{(\mu_k)_Y}  ,$$

where $\mu_k = P_X \circ W_k$, $P_X$ is the
optimal input distribution on $S$, and $\hat{P}_{(x_m, y)}$ is the
joint empirical distribution of the $m^{th}$ codeword $x_m$ and the
received word $y$. We denote this decoder by
$\text{GMAP}(W_1,\ldots,W_K)$. We also found that using the ML
metrics, instead of the MAP metrics $W_1,\ldots,W_K$, i.e.
$\text{GLRT}(W_1,\ldots,W_K)$, is not capacity achieving. \\

\begin{figure}
\begin{center}
\includegraphics[scale=.75]{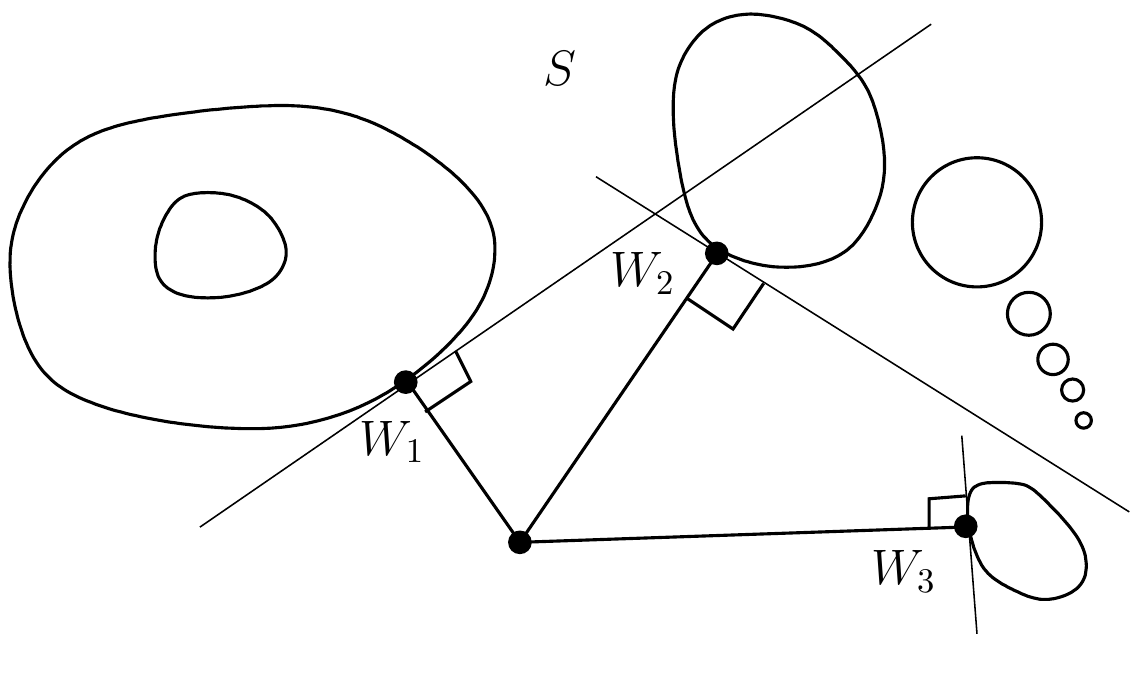}
\caption{GMAP with worst channels algorithm: here $S$ is represented by the union of all sets appearing in the figure. In this set, there are however only three one-sided components with respective worst channels $W_1, W_2$ and $W_3$, hence, decoding with the generalized linear decoder induced by the three corresponding MAP metrics is capacity achieving. MMI instead would have required an optimization of infinitely many metrics given by all possible DMC's.}
\label{algo}
\end{center}
\end{figure}

It is instrumental to compare our receiver with the MMI receiver. We observe that if the codeword
$x_m$ is chosen from a fixed composition $P_X$ code, the empirical
mutual information
\begin{eqnarray}
I(\hat{P}_{(x_m, y)})=\sup_{W } E_{\hat{P}_{(x_m, y)}} \log
\frac{W}{(P_X
\circ W)_Y} \label{mmiap}
\end{eqnarray}
where the maximization is taken over all possible DMC $W$, which means that
the MMI is actually the GMAP decoders taking into account all DMC's.
Our result says that we do not need to enumerate all DMC metrics to
achieve capacity, for a given compound set $S$, we can restrict
ourself to selecting carefully a subset of all metrics and yet
achieve the compound capacity. Those important metrics are found by
extracting the one-sided components of $S$, and taking the MAP
metrics induced by the worst channel of these components. When $S$
has a finite number of one-sided components, this decoder is
generalized linear.  The key step is to understand the structure of
the space of decoding metrics. The geometric insights gives rise to
a notion of which channels are dominated by which (with the
one-sided property) and how to combine the dominant representatives
of each components (Generalized MAP metrics).

We argued that the family of sets that can be written as finite
unions of one-sided sets covers a large variety of sets, even larger than the family of sets having finite unions of convex components. This means that the generalized linear
decoders with finitely many metrics can be found to achieve capacity
for a large family of compound sets. Yet, there do exist compound sets that are not even a finite union of
one-sided components. To see this, we can go back to the local 
geometric picture and imagine a compound set with infinitely many worst channels, for which the procedure shown in Figure \ref{algo} has to go through
an infinite number of steps. We argue, however, that such examples
are pedagogical, in the sense that if one is willing to give up a
small fraction of the capacity, then a finite collection of linear
decoding metrics would suffice. Moreover, there is a graceful
tradeoff between the number of metrics used, and the loss in
achievable rate.

Even more interestingly, one can develop a notion of a "blind"
generalized linear decoder, which does not even require the knowledge of
the compound set, yet guarantees to achieve a fraction of the
compound capacity. We describe here such decoders in the VN setting. As illustrated in Figure \ref{poly}, such decoders are induced by a set of metrics chosen in a "uniform" fashion. For a given compound set, we can then grow a polytope whose faces are the hyperplane orthogonal to these metrics and there will be a largest such polytope, that contains the entire compound set in its complement. This determines the rate that can be achieved with such a decoder on a given compound set, cf. $C_{\mathrm{poly}}$ in Figure \ref{poly}. In general $C_{\mathrm{poly}}$ is strictly less than the compound capacity, denoted by $C$ in Figure \ref{poly}; the only cases where $C=C_{\mathrm{poly}}$ is if by luck, one of the uniform direction is along the worst channel (and if there are enough metrics to contain the whole compound set).
Now, for a number $K$ of metrics, no matter what the compound set looks like, and not matter what its capacity is, the ratio between $C_{\mathrm{poly}}$ and $C$ can be estimated: in the VN geometry, this is equivalent to picking a sphere with radius $C$ and to compute the ratio between $C$ and the ``inner radius'' of a K-polytope inscribed in the sphere.
It is also clear that the higher the number of metrics is, the closer $C_{\mathrm{poly}}$ to $C$ is, and this controls the tradeoff between the computational complexity and the achievable rate. Again, as suggested by the very noisy picture, there is a graceful tradeoff between the number of metrics used, and the loss in achievable rate.
\begin{figure}
\begin{center}
\includegraphics[scale=.75]{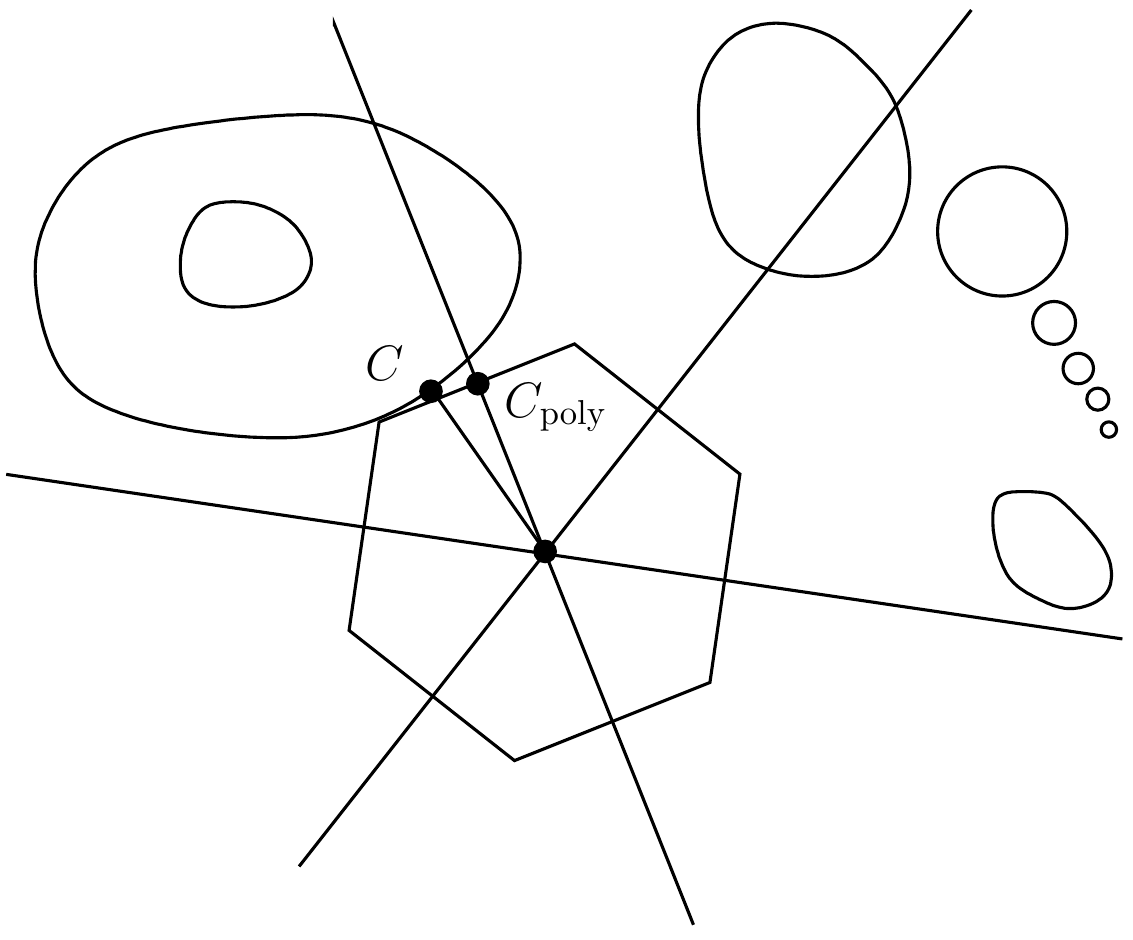}
\caption{A ``blind" generalized linear decoder for VN 3-ary compound channels, with 3 metrics chosen uniformly. The hexagon drawn in the figure is the largest hexagon defined by those uniform metrics that contains the compound set in its complement. This gives the achievable rate with such a decoder, namely $C_{\mathrm{poly}}$ in the figure, whereas the compound capacity is given by the minimum squared norm in the set, i.e. $C$ in the figure.}
\label{poly}
\end{center}
\end{figure}

\section{Conclusion}\label{conclud}

Many Information Theoretic problems evaluate the limiting
performance of a communication scheme by an expression optimizing
divergences under constrained probability distributions. The
divergence is not a formal distance, however, when the distributions
are close to each other, which we had by considering channels to be
very noisy, we are able to make local computations and the
divergence can be approximated by a squared norm. We showed that the
geometry governing this local setting is the one of an inner product
space, where notions of angles and distances are well defined. This
geometric insight simplifies greatly the problems. Rather than
getting a good approximation per-se, it provides a simplified
problem, for which we have a better insight and which points out
solutions to the original problem. It is also a powerful tool for
finding counter-examples. Finally, we showed how in this problem, we
could ``lift'' the results proven locally to results proven
globally.

\section*{Acknowledgment}
The authors wish to thank Emre Telatar, for helpful comments and stimulating discussions.

\bibliography{ludsubmitted}
\bibliographystyle{plain}

\end{document}